\documentclass[a4paper,10pt]{scrartcl}
\usepackage[utf8x]{inputenc}

\title{Representability Conditions\\ by Grassmann Integration}
\author{Volker Bach\footnote{Email: v.bach@tu-bs.de},\\
{\small{\em{Technische Universit\"at Braunschweig, Institut f\"ur Analysis und Algebra,}}}\\
{\small{\em{Rebenring 31, 38106 Braunschweig, Germany}}}\\
Hans Konrad Kn\"orr\footnote{Email: hanskonrad.knoerr@fernuni-hagen.de}\\
{\small{\em{FernUniversit\"at in Hagen, Fakult\"at f\"ur Mathematik und Informatik,}}}\\
{\small{\em{Lehrgebiet Angewandte Stochastik, 58084 Hagen, Germany}}}\\
Edmund Menge\footnote{Email: e.menge@tu-bs.de}\\
{\small{\em{Technische Universit\"at Braunschweig, Institut f\"ur Analysis und Algebra,}}}\\
{\small{\em{Rebenring 31, 38106 Braunschweig, Germany}}}}

\usepackage[english]{babel}
\usepackage[english, num]{isodate}
\usepackage{amsfonts,amsmath,amssymb,bbm}
\usepackage{amsthm}
\usepackage{nicefrac}
\usepackage{pifont}
\usepackage{yfonts}
\usepackage{geometry} 
\usepackage{graphicx}
\usepackage{booktabs}  
\usepackage{array} 
\usepackage{paralist} 
\usepackage{verbatim} 
\usepackage{subfig} 
\usepackage[nottoc,notlof,notlot]{tocbibind} 
\usepackage[titles,subfigure]{tocloft}  
\usepackage{caption}
\usepackage{mathtools}

\newcommand{\cB}{\mathcal{B}}

\newcommand{\cD}{\mathcal{D}}
\newcommand{\cE}{\mathcal{E}}

\newcommand{\cG}{\mathcal{G}}
\newcommand{\cH}{\mathcal{H}}

\newcommand{\cL}{\mathcal{L}}         

\newcommand{\cR}{\mathcal{R}}
\newcommand{\cS}{\mathcal{S}}

\def\e{{c}^{*}}
\def\v{{c}}

\newcommand{\TRh}[1]{\mathrm{tr}_1\left\{#1\right\}}
\newcommand{\TRhh}[1]{\mathrm{tr}_2\left\{#1\right\}}

\newcommand{\opsi}[1]{\overline{\psi}}
\newcommand{\oPSI}[1]{\overline{\Psi}}
\newcommand{\ophi}[1]{\overline{\phi}}
\newcommand{\oPHI}[1]{\overline{\Phi}}

\begin{document}

\theoremstyle{plain}
\newtheorem{thm}{Theorem}[section]
\newtheorem{lem}[thm]{Lemma}
\newtheorem{cor}[thm]{Corollary}
\theoremstyle{definition}
\newtheorem{defn}[thm]{Definition}
\theoremstyle{definition}
\newtheorem{rem}[thm]{Remark}
\bibliographystyle{plain}

\maketitle

\begin{abstract}
Representability conditions on the one- and two-particle density matrix for fermion systems are formulated by means of Grassmann integrals. A positivity condition for a certain kind of Grassmann integral is established which, in turn, induces the well-known
G-, P- and Q-Conditions of quantum chemistry by an appropriate choice of the integrand. Similarly, the $\mathrm{T}_1$- and $\mathrm{T}_2$-Conditions are derived.
Furthermore, quasifree Grassmann states are introduced and, for
every operator $\widetilde{\gamma}\in\cH\oplus\cH$ with $0\leq \widetilde{\gamma} \leq \mathbbm{1}$, the existence of a
unique quasifree Grassmann state whose one-particle density
matrix is $\widetilde{\gamma}$ is shown.
\end{abstract}


\section{Introduction}
The grand canonical energy (minus pressure) $E_0\left(\mu\right):=\inf\left\{\sigma\left\{\widehat{\mathbbm{H}}-\mu\widehat{\mathbbm{N}}\right\}\right\}$ at sufficiently large chemical potential $\mu\geq 0$ of a quantum 
system with a Hamiltonian $\widehat{\mathbbm{H}}$ and particle number operator $\widehat{\mathbbm{N}}$ is given by the Rayleigh--Ritz principle as
\begin{align}
 E_0\left(\mu\right)=\inf\left\{\mathrm{Tr}\left\{\rho^\frac{1}{2}\big(\widehat{\mathbbm{H}}-\mu\widehat{\mathbbm{N}}\big)\rho^\frac{1}{2}\right\}
     \Big|\ \rho\in DM\right\}\label{var},
\end{align}
where $\widehat{\mathbbm{H}}=\widehat{\mathbbm{H}}^*$ is a self-adjoint operator obeying stability of matter, i.e., bounded below by $-c\widehat{\mathbbm{N}}$ for
some $c<\infty$, and being at most quartic in the creation and annihilation operators \cite{ELT, WTH}. This is typically the case for models of non-relativistic
matter in physics and chemistry. The Pauli principle plays a crucial role for stability of matter to hold true, and we
thus restrict our attention to fermion systems. On the fermion Fock space $\wedge\cH$, the variation on the r.h.s. of~(\ref{var}) is over the set
\begin{align*}
 DM:=\left\{\rho\,\Big|\ \rho\in\cL_+^1\left(\wedge\cH\right),\ \mathrm{Tr}\left\{\rho\right\}=1,\ \big<\widehat{\mathbbm{N}}^2\big>_\rho<\infty\right\},
\end{align*}
i.e., density matrices with finite particle number variance. Here, the expectation value of an observable $\widehat{\mathbbm{A}}$ is 
\begin{align*}
 \big<\widehat{\mathbbm{A}}\big>_\rho:=\mathrm{Tr}\left\{\rho^\frac{1}{2}\widehat{\mathbbm{A}}\rho^\frac{1}{2}\right\}.
\end{align*}
More specifically, if
\begin{align*}
 \widehat{\mathbbm{H}}-\mu\widehat{\mathbbm{N}}=\sum\limits_{k,m}h_{km}\e(f_k)\v(f_m)+
	      \sum\limits_{k,l,m,n}V_{klmn}\e(f_l)\e(f_k)\v(f_m)\v(f_n),
\end{align*}
then 
\begin{align}
 E_0\left(\mu\right)=\inf\left\{\cE\left(\gamma_\rho,\Gamma_\rho\right)|\ \rho\in DM\right\}, \label{pdmgs}
\end{align}
where
\begin{align*}
 \cE\left(\gamma_\rho,\Gamma_\rho\right)=\sum\limits_{k,m}h_{km}\left<f_m, \gamma_\rho f_k\right>
      +\sum\limits_{k,l,m,n}V_{klmn}\left<f_m\otimes f_n, \Gamma_\rho\left(f_k\otimes f_l\right)\right>
\end{align*}
and the one- and two-particle density matrices corresponding to $\rho$ are defined by
\begin{align*}
 \left<f,\gamma_\rho g\right>&:=\left<\e(g)\v(f)\right>_\rho\quad\text{and}\\
 \left<f\otimes g,\Gamma_\rho\big(\tilde{f}\otimes \tilde{g}\big)\right>&:=\left<\e(\tilde{g})\e\big(\tilde{f}\big)
	\v(f)\v(g)\right>_\rho,
\end{align*}
respectively, for all $f,g,\tilde{f},\tilde{g}\in\cH$.
Note that (\ref{pdmgs}) can be rewritten as
\begin{align}
 E_0\left(\mu\right)=\inf\left\{\cE\left(\gamma,\Gamma\right)|\ \left(\gamma,\Gamma\right)\in\cR\right\}, \label{pdmgs2}
\end{align}
where
\begin{align*}
 \cR:=\left\{\left(\gamma,\Gamma\right)\in\cL^1\left(\cH\right)\times\cL^1\left(\cH\otimes\cH\right)\Big|\ 
	\exists \rho\in DM:\ \left(\gamma,\Gamma\right)=\left(\gamma_\rho,\Gamma_\rho\right)\right\}
\end{align*}
denotes the set of all representable one- and two-particle density matrices. 
Equation (\ref{pdmgs2}) suggests that the search for a minimizing $\rho$ could be drastically simplified if one would find 
a characterization of all representable reduced density matrices $\left(\gamma,\Gamma\right)$. This was realized almost
fifty years ago \cite{AJC, RME, CJP, POL}, but such a characterization is still unknown.

The characterization of $E_0\left(\mu\right)$ by (\ref{pdmgs2}) immediately yields lower bounds of the form
\begin{align}
 E_0\left(\mu\right)=: E_{\cR}\left(\mu\right)\geq E_{\cS}\left(\mu\right), \label{inequu}
\end{align}
for any superset $\cS$ of $\cR$.
For example, the positivity $\left<P_2^*P_2\right>_\rho\geq 0$ for all polynomials $P_2\equiv P_2\left(\e,\v\right)$ in the
creation and annihilation operators of degree two yields the so-called G-, P-, and Q-Conditions on $\left(\gamma_\rho,\Gamma_\rho\right)$
\cite{BKM, AJC, RME, CJP}. Similarly, the positivity $\left<P_3^* P_3+P_3 P_3^*\right>_\rho\geq 0$ yields the $\mathrm{T}_1$-
and generalized $\mathrm{T}_2$-Conditions \cite{RME}. Hence, all representable reduced density matrices $\left(\gamma,\Gamma\right)$
necessarily fulfill the G-, P-, Q-, $\mathrm{T}_1$-, and generalized $\mathrm{T}_2$-Conditions, and we have
\begin{align}
 E_\cR\left(\mu\right)\geq E_{\cS\left[\mathrm{G,P,Q,}\mathrm{T}_1,\mathrm{T}_2\right]}\left(\mu\right)\geq E_{\cS\left[\mathrm{G,P,Q}\right]}\left(\mu\right),  \label{inequuu}
\end{align}
since $\cR\subseteq \cS\left[\mathrm{G,P,Q,}\mathrm{T}_1,\mathrm{T}_2\right]\subseteq \cS\left[\mathrm{G,P,Q}\right]$, with
\begin{align*}
 \cS\left[X\right]:=\left\{\left(\gamma,\Gamma\right)\in\cL^1\left(\cH\right)\times\cL^1\left(\cH\otimes\cH\right)\Big|\ 
      \left(\gamma,\Gamma\right)\ \mathrm{fulfills}\ \mathrm{Conditions}\ X\right\}.
\end{align*}
We have discussed (\ref{inequu})-(\ref{inequuu}) for $\cS=\cS\left[\mathrm{G, P}\right]$ in some detail in \cite{BKM} and refer the reader to that paper 
and references therein. Furthermore, for $\cS=\cS\left[\mathrm{G,P,Q,}\mathrm{T}_1,\mathrm{T}_2\right]$ numerical works
show agreement with Full CI computations \cite{CLS, DAM, MAZ, ZBF} to high accuracy.\\

The purpose of the present paper is the reformulation of representability conditions in terms of Grassmann integrals. 
Such a transcription may possibly yield new viewpoints and hopefully new insights into the representability problem.
To this end, we introduce a Grassmann algebra $\cG_M$ as a finite dimensional complex algebra. The object on $\cG_M$ 
corresponding to a given density matrix is an element of the form $\vartheta^*\star\vartheta$ described in the sequel. 
Grassmann integration is the basic and most commonly used method (see, e.g., \cite{FKT, MSH}) in theoretical physics
to compute partition functions of the form 
\begin{align*}
 Z_{\Gamma,\lambda}\left(J\right):=\int \mathrm{D}_\Gamma(\phi)\,\mathrm{e}^{-S_\Gamma+\left(J,\phi\right)_\Gamma}
\end{align*}
as a functional integral with $\mathrm{D}_\Gamma(\phi):=\prod\limits_{x\in\Gamma}\mathrm{d}\phi\left(x\right)$ with sources $J:\Gamma\to\mathbbm{R}$
and an action $S_\Gamma$ (see \cite{MSH} for further details).\\

The derivation of the G-, P-, Q-, $\text{T}_1$-, and generalized $\text{T}_2$-Conditions
is based on the representation of the trace on $\wedge\cH$ in terms of Grassmann integrals and a positivity condition of a Grassmann integral, namely
\begin{align}
 \forall \eta \in \cG_M:\ \int\mathrm{d}(\overline{\Psi},\Psi)\,\mathrm{e}^{2\left(\overline{\Psi},\Psi\right)}\eta^*\star\eta \geq 0, \label{mugl}
\end{align}
where $\int\mathrm{d}(\overline{\Psi},\Psi)$ denotes the Grassmann integration. The star product refers to a product
on $\cG_M$ and is introduced later. Considering appropriate subspaces of $\cG_M$ denoted by $\cG_M^{(n)}$, the main 
results of this paper are the bounds for the one-particle density matrix $\gamma_\vartheta$:
 \begin{align*}
 \left\{\forall \mu\in\cG_M^{(1)}: \int\mathrm{d}(\overline{\Psi},\Psi)\,\mathrm{e}^{2\left(\overline{\Psi},\Psi\right)}\vartheta^*\star\vartheta \star\mu\geq 0 \right\}
    \Leftrightarrow \left\{0\leq \gamma_\vartheta\leq \mathbbm{1}\right\}
 \end{align*}
and the G-, P-, and Q-Condition as conditions for the two-particle density matrix $\Gamma_\vartheta$:
\begin{align*}
\nonumber &\left\{\forall \mu\in\cG_M^{(2)}: \int\mathrm{d}(\overline{\Psi},\Psi)\,\mathrm{e}^{2\left(\overline{\Psi},\Psi\right)}\vartheta^*\star\vartheta \star\mu\geq 0 \right\}   \Leftrightarrow \left\{0\leq \gamma_\vartheta\leq \mathbbm{1},\ \text{G-, P-, and Q-Condition} \right\}
\end{align*}
Finally, we prove the validity of the $\text{T}_1$- and generalized $\text{T}_2$-Condition by Inequality (\ref{mugl}).


\section{Reduced Density Matrices and Representability}

Before we elucidate how to derive the G-, P-, Q-, $\mathrm{T}_1$-, and generalized $\mathrm{T}_2$-Conditions for the 1- and 2-particle density matrix
(1- and 2-pdm) by Grassmann integration, we give a definition of these first two reduced density matrices. For this purpose, we 
consider a finite-dimensional index set $M$, an $|M|$-dimensional (one-particle) Hilbert space $\left(\cH,\ \left<\,\cdot\, ,\,\cdot\,\right>\right)$, and an arbitrary, but fixed orthonormal basis (ONB) 
$\left\{\psi_i\right\}_{i\in M}$ of $\cH$.
Furthermore, we introduce the usual fermion creation and annihilation operators on the fermion Fock space $\wedge\cH$ over $\cH$ given by 
$\e(\psi_i)\equiv\e_i$ and $\v(\psi_i)\equiv\v_i$ with the canonical anticommutation relations (CAR)
\begin{align}
 \left\{\v(f),\v(g)\right\} = \left\{\e(f),\e(g)\right\} = 0\ \ \text{and}\
  \left\{\v(f),\e(g)\right\}=\left<f,g\right>\cdot\mathbbm{1}_{\wedge\cH}\label{DefCAR}
\end{align}
for all $f,g\in\cH$, where $\left<\,\cdot\,,\,\cdot\,\right>$ is linear in the second and antilinear in the first argument. $\left\{A,B\right\}:=AB+BA$
denotes the anticommutator.\\

The 1-pdm $\gamma_\rho\in\cL_{+}^{1}\left(\cH\right)$ of a density matrix $\rho$, i.e., a positive trace class operator on $\wedge\cH$ of unit trace
($\mathrm{tr}_{\wedge\cH}\left\{\rho\right\}=1$), is defined by its matrix elements as
\begin{align}
 \forall f,g\in\cH:\ \left<f ,\gamma_\rho g\right>&:=\text{tr}_{\wedge\cH}\left\{\rho\, \e(g)
	      \v(f)\right\}\label{1pdm}.
\end{align}
Likewise, the 2-pdm $\Gamma_\rho\in\cL_{+}^{1}\left(\cH\otimes\cH\right)$ of $\rho$ is defined by
\begin{align}
 \forall f_1, f_2 ,g_1 , g_2 \in \cH:\ \left<f_1\otimes f_2 , \Gamma_\rho( g_1\otimes g_2)\right>
:=\text{tr}_{\wedge\cH}\left\{\rho\,\e(g_2)\e(g_1)\v(f_1)\v(f_2)\right\}. \label{2pdm}
\end{align}

There are several properties which can be derived directly from the definition of $\gamma_\rho$ and $\Gamma_\rho$.
\begin{lem}
 Let $\rho\in\cL_{+}^{1}(\wedge\cH)$ be a density matrix and $\widehat{\mathbbm{N}}:=\sum\limits_{k\in M}\e_k\v_k$ the particle number
operator  with $\big<\widehat{\mathbbm{N}}^2\big>_\rho<\infty$. Then the following assertions hold true:
\begin{itemize}
\item[i)]$\gamma_\rho\in\cL_{+}^{1}(\cH)$,\quad $0\leq\gamma_\rho\leq \mathbbm{1}$,\quad $\mathrm{tr}_{\cH}\left\{\gamma_\rho\right\}
=\big<\widehat{\mathbbm{N}}\big>_\rho$,\quad $\Gamma_\rho\in\cL_{+}^{1}(\cH\otimes\cH)$,\quad $0\leq
\Gamma_\rho\leq \big<\widehat{\mathbbm{N}}\big>_\rho$,\quad and\quad $\mathrm{tr}_{\cH\otimes\cH}\left\{\Gamma_\rho\right\}
= \big<\widehat{\mathbbm{N}}\big(\widehat{\mathbbm{N}}-\mathbbm{1}\big)\big>_\rho$.
\item[ii)]
If $\mathrm{Ran}\left\{\rho\right\}\subseteq\wedge^{(N)}\cH$, $N\in\mathbbm{N}$, then, for all $f,g\in\cH$,
	\begin{align*}
	 \left<f,\gamma_\rho g\right>= \frac{1}{N-1}\sum\limits_{k\in M}\left<f\otimes\varphi_k 
	  ,\Gamma_\rho(g\otimes\varphi_k)\right>,
	\end{align*}
	where $\left\{\varphi_k\right\}_{k\in M}\subseteq\cH$ is an ONB. Here, $\wedge^{(N)}\cH$ denotes the fermion $N$-particle Fock space.
\item[iii)] Furthermore,
	\begin{align*}
	\rho=\left|\e(\varphi_1)\cdots\e(\varphi_N)\Omega\right>
	\left<\e(\varphi_1)\cdots\e(\varphi_N)\Omega\right|\ \ \Leftrightarrow \ \ 
	\gamma_\rho=\sum\limits_{i=1}^{N}\left|\varphi_i\right>\left<\varphi_i\right|
	\end{align*}
	and, in this case,
	\begin{align*}
	 \Gamma_\rho=\left(\mathbbm{1}-\mathrm{Ex}\right)\left(\gamma_\rho\otimes\gamma_\rho\right),
	\end{align*}
	where $\mathrm{Ex}\left(f\otimes g\right):=g\otimes f$ for any $f,g\in\cH$.
\end{itemize}
\end{lem}

For further details we recommend \cite{VBA, BKM, AJC, CJP}. A proof can be found in \cite{VBA}. Beside these properties,
necessary conditions on $\left(\gamma,\Gamma\right)$ to be representable were derived in \cite{AJC, RME, CJP}. In particular,
the P-, G-, and Q-Conditions are:
\begin{itemize}
 \item $\left\{\left(\gamma,\Gamma\right)\ \mathrm{fulfills}\ \text{P-Condition}\right\} :\Leftrightarrow
   \left\{\Gamma\geq 0\right\},$
 \item $\nonumber \left\{\left(\gamma,\Gamma\right)\ \mathrm{fulfills}\ \text{G-Condition}\right\}\\ 
\qquad:\Leftrightarrow\left\{\forall A\in\cB\left(\cH\right):\ \mathrm{tr}\left\{\left(A^*\otimes A\right)\left(\Gamma+\mathrm{Ex}\left(\gamma\otimes\mathbbm{1}\right)\right)\right\}\geq \left|\mathrm{tr}\left\{A\gamma\right\}\right|^2\right\},$
 \item $\left\{\left(\gamma,\Gamma\right)\ \mathrm{fulfills}\ \text{Q-Condition}\right\}\\
:\Leftrightarrow\left\{\Gamma+\left(\mathbbm{1}-\mathrm{Ex}\right)\left(\mathbbm{1}\otimes\mathbbm{1}-\gamma\otimes\mathbbm{1}-\mathbbm{1}\otimes\gamma\right)\geq 0\right\}.$
\end{itemize}
The $\mathrm{T}_1$- and generalized $\mathrm{T}_2$-Conditions
are more complicated and not given here. For this conditions we refer the reader to \cite{RME} or Subsection 5.3 of this work.


\section{Grassmann Algebras}

We introduce the Grassmann algebra as the complex algebra generated by elements of the set
$\left\{\overline{\psi}_i,\psi_i\right\}_{i\in M}$ with $|M|<\infty$ modulo
the anticommutation relations specified below. A product of two generators is denoted
by $\psi_i\cdot\psi_j\equiv\psi_i\psi_j$. The unity is given as $1\cdot\psi_i=\psi_i\cdot 1=\psi_i$ (and
equivalently for $\overline{\psi}_j$).
The anticommutation relations allow us
to find a one-to-one representation of the CAR of fermion creation and annihilation
operators in terms of Grassmann variables. 
 For further details on this well-known material we 
recommend \cite{FRZ, WSP, MSH, LAT}. We use the notation of \cite{WSP}.
 
\begin{defn}
 For an ordered set $I:=\left\{i_1,\dots,i_m\right\}\subseteq M$ we write
\begin{align*}
 \Psi_{I}:=\psi_{i_1}\cdots\psi_{i_m},\quad \overline{\Psi}_{I}:=\overline{\psi}_{i_1}\cdots\overline{\psi}_{i_m}.
\end{align*}
For $I=\emptyset$ we set $\Psi_I=\overline{\Psi}_I=1$. Denoting the reversely ordered set corresponding to $I$ by $I'$, we write
\begin{align*}
 \Psi_{I'}:=\psi_{i_m}\cdots\psi_{i_1}.
\end{align*}
\end{defn}
\begin{defn}
Given a set of generators $\left\{\overline{\psi}_i,\psi_i\right\}_{i\in M}$ obeying the anticommutation relations
\begin{align*}
 \overline{\psi}_i\psi_j+\psi_j\overline{\psi}_i=\overline{\psi}_i\overline{\psi}_j+\overline{\psi}_j\overline{\psi}_i
  =\psi_i\psi_j+\psi_j\psi_i=0\quad \forall\,i,j\in M,
\end{align*}
the Grassmann algebra $\cG_M$ is defined as 
\begin{align*}
 \cG_M:=\mathrm{span}\left\{\overline{\Psi}_{I}\Psi_{J}|\ I,J\subseteq M\right\}.
\end{align*}
\end{defn}

Introducing the ordinary wedge product, we can identify $\cG_{M}$ with the Fock space $\wedge\left(\overline{\cH}\oplus\cH\right)$
of a Hilbert space $\left(\cH,\ \left<\,\cdot\, ,\,\cdot\,\right> \right)$ with finite dimension $|M|$. Considering $\cH$ as a subset of $\cG_M$, we can identify $\left\{\psi_i\right\}_{i\in M}$
with a fixed ONB of $\cH$ and $\left\{\overline{\psi}_i\right\}_{i\in M}$ with the corresponding ONB of  $\overline{\cH}$, i.e., the 
space of all continuous linear functionals $\cH\rightarrow\mathbbm{C},\ \psi_{i}\mapsto\overline{\psi}_i\left(\,\cdot\,\right):=\left<\psi_i,\cdot\,\right>$.
\begin{rem}
If $\cG_M$ is generated by $\left\{\overline{\phi}_i,\phi_i\right\}_{i\in M}$, we emphasize this 
by using $\mu\left(\overline{\phi},\phi\right)\in\cG_M$ instead of $\mu\in\cG_M$. We also use ``mixed'' generators, e.g.,
\begin{align*}
 \mu\left(\overline{\psi},\phi\right):=\sum\limits_{i,j}\alpha_{ij}\,
	  \overline{\Psi}_{I_i}\Phi_{J_j}.
\end{align*}
\end{rem}

Later, it is necessary to link the CAR algebra of fermion annihilation and creation operators to a Grassmann algebra.
For this purpose, a map between $\cB\left(\wedge\cH\right)$ and $\cG_{M}$ as an isomorphism between vector spaces is required. This map is provided below.
\begin{defn}
  Let $\cG_M$ be generated by $\left\{\overline{\psi}_i,\psi_i\right\}_{i\in M}$ and associate $\left\{\psi_i\right\}_{i\in M}$ with
 a fixed ONB of $\cH$. For all $z\in\mathbbm{C}$ and $m,n\leq |M|$, we define the linear map $\Theta:\mathcal{B}\left(\wedge\cH\right)\to\cG_{M}$
  by $\Theta\left(z\right):=z$ and
\begin{align}
 \Theta\left(\e(\psi_{i_1})\cdots\e(\psi_{i_m})\v(\psi_{j_1})
	  \cdots\v(\psi_{j_n})\right)
      :=\overline{\psi}_{i_1}\cdots\overline{\psi}_{i_m}\psi_{j_1}\cdots\psi_{j_n}, \label{deda}
\end{align}
and extension to $\cB\left(\wedge\cH\right)$ by linearity.
\end{defn}

We emphasize that $\Theta$ is not multiplicative. E.g., while
\begin{align*}
\Theta\big(\e(\psi_1)\v(\psi_1)\big)=\overline{\psi}_1\psi_1
=\Theta\big(\e(\psi_1)\big)\Theta\big(\v(\psi_1)\big),
\end{align*} 
we have
\begin{align*}\Theta\big(\v(\psi_1)\e(\psi_1)\big)
&=\Theta\big(-\e(\psi_1)\v(\psi_1)+\mathbbm{1}\big)\\
&=-\overline{\psi}_1\psi_1+1=\psi_1\overline{\psi}_1+1
=\Theta\big(\v(\psi_1)\big)\Theta\big(\e(\psi_1)\big)+1.
\end{align*}
Thus, Equation (\ref{deda}) only holds for normal-ordered monomials in creation and annihilation ope\-ra\-tors, i.e., monomials in which all 
creation ope\-ra\-tors are to the left of all annihilation ope\-ra\-tors.

\begin{defn}
  For any $A\in\mathcal{B}\left(\cH\right)$ we set
  \begin{align*}
   \left(\overline{\Psi},A\Phi\right):=\sum\limits_{i,j\in M}\left[\overline{\psi}_i\left(A\psi_j\right)\right]
      \overline{\psi}_j\phi_i\in\cG_M.
  \end{align*}
\end{defn}

Note that $\overline{\psi}_i\left(A\psi_j\right)=\left<\psi_i,A\psi_j\right> \in\mathbbm{C}$. Furthermore, $\left(\overline{\Psi},A\Phi\right)$ 
does not depend on the choice of generators of $\cG_M$ as can be seen by a unitary change of generators, e.g., 
$\chi_i:=\sum\limits_{j\in M}U_{ij}\psi_j$ for unitary $U$. An important case is $A=\mathrm{id}_\cH$. Here we have 
$\left(\overline{\Psi},\Phi\right)=\sum\limits_{i\in M}
      \overline{\psi}_i\phi_i$.
One of the last ingredients for the Grassmann integration is the following.
\begin{defn} The expression
 $\mathrm{e}^{\pm\left(\overline{\Psi},A\Phi\right)}\in\cG_M$
is given by
\begin{align*}
 \mathrm{e}^{\pm\left(\overline{\Psi},A\Phi\right)}:=\sum\limits_{m=0}^{\infty}\frac{1}{m!}
    \left[\pm \left(\overline{\Psi},A\Phi\right)\right]^m.
\end{align*}
\end{defn}

As $\mathrm{dim}\left\{\wedge\cH\right\}=2^{\mathrm{dim}\left\{\cH\right\}}$, the sum runs only over
$0\leq m\leq 2^{\mathrm{dim}\left\{\cH\right\}}$. 
\begin{rem}
 Since $\left(\overline{\Psi},\Phi\right)=\sum\limits_{\alpha\in M}\overline{\psi}_\alpha\phi_\alpha$, and 
$\overline{\psi}_\alpha\phi_\alpha$ commutes with every element of $\cG_M$, we have
\begin{align}
 \mathrm{e}^{\pm\left(\overline{\Psi},\Phi\right)}=
      \prod\limits_{\alpha\in M}\left(1\pm\overline{\psi}_\alpha\phi_\alpha\right). \label{expfkt}
\end{align}
\end{rem}
\begin{defn}
 For all $i,j\in M$, we define the vector space homomorphisms
$
 \frac{\delta}{\delta \psi_i},\ \frac{\delta}{\delta \overline{\psi}_i}: \cG_M
      \to \cG_M $
by 
\begin{align*}
 \frac{\delta}{\delta \psi_i}\psi_j=\frac{\delta}{\delta \overline{\psi}_i}\overline{\psi}_j=\delta_{ij}, \quad\mathrm{and}\quad 
  \frac{\delta}{\delta \psi_i}\overline{\psi}_j=\frac{\delta}{\delta \overline{\psi}_i} \psi_j=0.
\end{align*}
\end{defn}
\begin{rem}The set $\left\{\frac{\delta}{\delta \overline{\psi}_i},\frac{\delta}{\delta \psi_i}\right\}_{i\in M}$ itself generates a 
Grassmann algebra.
\end{rem}


\section{Grassmann Integration}

Now we are prepared to define the Grassmann integral, which is a linear operator from 
$\cG_M$ to $\mathbbm{C}$.
\begin{defn}
The map $\int\mathrm{d}(\overline{\Psi},\Psi):\ \cG_M\to \mathbbm{C}$
is defined by 
\begin{align*}
 \int\mathrm{d}(\overline{\Psi},\Psi):=
\prod\limits_{\alpha\in M}\left(\frac{\delta}{\delta \overline{\psi}_\alpha}\frac{\delta}{\delta \psi_\alpha}\right).
\end{align*}
and is referred to as the Grassmann integral.
\end{defn}
\begin{rem}
 If the factor $\mathrm{e}^{2\left(\overline{\Psi},\Psi\right)}=\prod\limits_{\alpha\in M}\left(1+2\overline{\psi}_\alpha\psi_\alpha\right)$
is involved in the integration, we use the abbreviation
\begin{align*}
\cD(\overline{\Psi},\Psi):= \mathrm{d}(\overline{\Psi},\Psi)\,\mathrm{e}^{2\left(\overline{\Psi},\Psi\right)},
\end{align*}
since $\prod\limits_{\alpha\in M}\left(1+2\overline{\psi}_\alpha\psi_\alpha\right)$ commutes with every element of $\cG_M$.
\end{rem}

In order to state the invariance of the Grassmann integration with respect to a change of generators,
we introduce some notations. We write two sets of generators, $\left\{\overline{\psi}_i,\psi_i\right\}_{i\in M}$ and $\left\{\overline{\chi}_i,\chi_i\right\}_{i\in M}$, as $2|M|$-component
vectors $\underline{a}$ and $\underline{b}$, respectively, whose entries are given by 
\begin{align}
 a_i:=\overline{\psi}_i\ \text{and}\ a_{|M|+i}:=\psi_i, \quad\text{and}\quad 
    b_i:=\overline{\chi}_i\ \text{and}\ b_{|M|+i}:=\chi_i. \label{defvecab}
\end{align}
where for all $i\in M$. Furthermore, we define the entries of the $2|M|$-component vectors $\frac{\delta}{\delta\underline{a}}$ and $\frac{\delta}{\delta\underline{b}}$ by
\begin{align*}
\frac{\delta}{\delta\underline{a}_{i}}:=\frac{\delta}{\delta \overline{\psi}_i}\ \text{and}\ \frac{\delta}{\delta\underline{a}_{|M|+i}}:=\frac{\delta}{\delta \psi_i},
\quad\text{and}\quad 
\frac{\delta}{\delta\underline{b}_{i}}:=\frac{\delta}{\delta \overline{\chi}_i}\ \text{and}\ \frac{\delta}{\delta\underline{b}_{|M|+i}}:=\frac{\delta}{\delta \chi_i}.
\end{align*}
We denote the index set for the introduced vectors by $\widetilde{M}$, $|\widetilde{M}|=2|M|$. In this notation the Grassmann integration with respect to $\left\{\overline{\psi}_i,\psi_i\right\}_{i\in M}$ reads as 
\begin{align*}
 \left(-1\right)^{\frac{1}{2}|M|(|M|-1)}\prod\limits_{\alpha\in M}\left(\frac{\delta}{\delta \overline{\psi}_\alpha}\frac{\delta}{\delta \psi_\alpha}\right)
      =\prod\limits_{\alpha\in M}\frac{\delta}{\delta \overline{\psi}_\alpha}
	  \prod\limits_{\alpha\in M}\frac{\delta}{\delta \psi_\alpha}
      =\prod\limits_{\beta\in \widetilde{M}}\frac{\delta}{\delta\underline{a}_{\beta}}.
\end{align*}

\begin{lem}
The Grassmann integral does not depend on the choice of the ge\-ne\-ra\-tors. I.e., for $\underline{a}$ and $\underline{b}$ as defined
in (\ref{defvecab}), and a transformation defined by  
\begin{align*}
 \underline{b}=U\,\underline{a},
\end{align*}
where $U$ is a unitary $2|M|\times 2|M|$-matrix, we have
\begin{align*}
 \frac{\delta}{\delta \underline{b}}= \overline{U}\frac{\delta}{\delta \underline{a}}
\end{align*}
 and, for any $\mu\in\cG_M$,
\begin{align*}
 \prod\limits_{\alpha\in M}\left(\frac{\delta}{\delta \overline{\psi}_\alpha}\frac{\delta}{\delta \psi_\alpha}\right)
      \mu\left(\overline{\psi},\psi\right)
    =\prod\limits_{\alpha\in M}\left(\frac{\delta}{\delta \overline{\chi}_\alpha}\frac{\delta}{\delta \chi_\alpha}\right)
	\mu\left(\overline{\chi},\chi\right).
\end{align*}
\end{lem}
\begin{proof}
First we prove $\frac{\delta}{\delta \underline{b}}= \overline{U}\frac{\delta}{\delta \underline{a}}$. The identity
$\frac{\delta}{\delta a_j} a_i = \delta_{ij}$ follows from the properties of the generators. An equivalent identity has to be claimed for $\frac{\delta}{\delta \underline{b}}\underline{b}$.
Suppose $\frac{\delta}{\delta \underline{b}}$ transforms as $\frac{\delta}{\delta \underline{b}}=V\frac{\delta}{\delta \underline{a}}$ with a $2|M|\times 2|M|$-matrix $V$.
This yields
\begin{align*}
 \frac{\delta}{\delta b_j}b_i &=\Bigg(\sum\limits_{\alpha\in \widetilde{M}}V_{j\alpha}\frac{\delta}{\delta a_\alpha}\Bigg)\Bigg(\sum\limits_{\beta\in \widetilde{M}}U_{i\beta}a_\beta\Bigg)=\left(UV^T\right)_{ij}.
\end{align*}
In other words, we have $UV^T=\mathbbm{1}_{\widetilde{M}}$ and, thus, $V=\overline{U}$.
Finally, we can prove the invariance of the Grassmann integral. For a given set of generators $\left\{\overline{\psi}_i,\psi_i\right\}_{i\in M}$, any $\mu\in\cG_M$ can be written as 
\begin{align*}
 \mu\equiv\mu\left(\overline{\psi},\psi\right)=\sum\limits_{I,J\subseteq M}\alpha_{IJ}\overline{\Psi}_I\Psi_J,
\end{align*}
where $\alpha_{IJ}\in\mathbbm{C}$ for all $I,J\subseteq M$, and $I,J$ ordered. The Grassmann integral of $\mu$ is
\begin{align*}
 \int\mathrm{d}(\overline{\Psi},\Psi)\,\mu\left(\overline{\psi},\psi\right)=\int\mathrm{d}(\overline{\Psi},\Psi)\sum\limits_{I,J\subseteq M}a_{IJ}\overline{\Psi}_I\Psi_J
      =\int\mathrm{d}(\overline{\Psi},\Psi)\,\alpha_{MM}\overline{\Psi}_M\Psi_M,
\end{align*}
since all other terms of $\mu$ do not contribute to the integral. If the decomposition of $\mu$ yields $\alpha_{MM}=0$, the Grassmann integral
of $\mu$ vanishes. In this case there is nothing to show. For $\alpha_{MM}\neq 0$ we consider the transformation of 
$\int\mathrm{d}(\overline{\Psi},\Psi)$ and $\overline{\Psi}_M\Psi_M$ separately. For $\int\mathrm{d}(\overline{\Psi},\Psi)$ we use
$\frac{\delta}{\delta\underline{a}_{i}}\frac{\delta}{\delta\underline{a}_{j}}=-\frac{\delta}{\delta\underline{a}_{j}}\frac{\delta}{\delta\underline{a}_{i}}$ for 
$i\neq j$, and express $\frac{\delta}{\delta\underline{b}}$ in terms of 
    $\frac{\delta}{\delta\underline{a}}$:
\begin{align*}
\left(\prod\limits_{\alpha\in M}\frac{\delta}{\delta\overline{\chi}_\alpha}\right)\left(\prod\limits_{\alpha\in M}\frac{\delta}{\delta\chi_\alpha}\right)&=\prod\limits_{\beta\in \widetilde{M}}\frac{\delta}{\delta\underline{b}_{\beta}}
    =\sum\limits_{\beta_1,\dots,\beta_{|\widetilde{M}|}\in \widetilde{M}}\prod\limits_{j\in \widetilde{M}}\overline{U}_{j \beta_j}\frac{\delta}{\delta\underline{a}_{\beta_j}}\\
    &=\sum\limits_{\pi\in\cS_{\widetilde{M}}}\prod\limits_{j\in \widetilde{M}}\overline{U}_{j \pi(j)}\frac{\delta}{\delta\underline{a}_{\pi(j)}}
    =\sum\limits_{\pi\in\cS_{\widetilde{M}}}\left(-1\right)^\pi \prod\limits_{j\in \widetilde{M}}\overline{U}_{j \pi(j)}\frac{\delta}{\delta\underline{a}_{j}}\\
    &=\mathrm{det}\left(\overline{U}\right)\prod\limits_{j\in \widetilde{M}}\frac{\delta}{\delta\underline{a}_{j}}.
\end{align*}
Analogously, we have 
\begin{align*}
 \prod\limits_{\alpha\in M}\overline{\chi}_M\prod\limits_{\alpha\in M}\chi_M=\prod\limits_{\beta\in \widetilde{M}}b_\beta
	=\mathrm{det}\left(U\right)\prod\limits_{j\in \widetilde{M}}a_j.
\end{align*}
Merging the results we obtain 
\begin{align*}
 \left(\prod\limits_{\alpha\in M}\frac{\delta}{\delta\overline{\chi}_\alpha}\right)\left(\prod\limits_{\alpha\in M}\frac{\delta}{\delta\chi_\alpha}\right)
      \prod\limits_{\alpha\in M}\overline{\chi}_M\prod\limits_{\alpha\in M}\chi_M=\left|\mathrm{det}\left(U\right)\right|^2\prod\limits_{j\in \widetilde{M}}\frac{\delta}{\delta\underline{a}_{j}}\prod\limits_{j\in \widetilde{M}}a_j.
\end{align*}
The proof is complete with $\left|\mathrm{det}\left(U\right)\right|^2=1$, since $U$ is unitary.
\end{proof}
\begin{rem}
 The transformation $U$ mixes $\overline{\psi}_i$'s and $\psi_i$'s. For 
 $U:=\begin{pmatrix} u & v \\ \overline{v} & \overline{u}\end{pmatrix}$, a
 transformation without mixing is given for $v=0$. In this case, $u$ has to be unitary.
\end{rem}

For the application of the Grassmann integration on representability conditions we still need some tools, 
especially the definition of a product on $\cG_M$ which induces the CAR on the Grassmann algebra. 
\begin{defn}
 For all $\mu\equiv\mu\left(\overline{\psi},\psi\right)$ and $\eta\equiv\eta\left(\overline{\psi},\psi\right)\in\cG_M$, we define the star product 
$\mu\star \eta\in \cG_M$ by
\begin{align*}
&\left(\mu\star \eta\right)\left(\overline{\psi},\psi\right):=\int\mathrm{d}(\overline{\Phi},\Phi)\,
      \mu\left(\overline{\psi},\phi\right)\eta\left(\overline{\phi},\psi\right)
      \mathrm{e}^{-\left(\overline{\Psi},\Psi\right)}
      \mathrm{e}^{\left(\overline{\Psi},\Phi\right)}
      \mathrm{e}^{-\left(\overline{\Phi},\Phi\right)}
      \mathrm{e}^{\left(\overline{\Phi},\Psi \right)}.
\end{align*}
\end{defn}

We calculate the star product of two monomials $\mu:=\overline{\Psi}_ {I}\Psi_{J}$ and $\eta:=\overline{\Psi}_{K}\Psi_{L}$,
which determines the star product in general, due to the linearity of the Grassmann integral.
\begin{lem}\label{basicstar}
 Let $I,J,K,L\subseteq M$. Then we have
\begin{align}
\nonumber &\left(\overline{\Psi}_{I}\Psi_{J}\right)\star\left(\overline{\Psi}_{K}\Psi_{L}\right)\\
  &\qquad\qquad=\sigma_S\sigma_{JS}\cdot\mathrm{e}^{-\left(\overline{\Psi},\Psi\right)}
    \overline{\Psi}_{I}\Psi_{J\backslash S}\overline{\Psi}_{K\backslash S}\Psi_{L}
		\prod\limits_{\substack{\alpha\in M\\ \backslash\left(J\cup K\right)}}\left(1+\overline{\psi}_\alpha
	  \psi_\alpha\right), \label{starprod}
\end{align}
where $S:= J\cap K$ and $\sigma_{JS}:=\left(-1\right)^{|S|\left(|J\backslash S|+\frac{|S|-1}{2}\right)}$. The sign
$\sigma_S$ is given by the identity $\sigma_S\Phi_{S}\Phi_{J\backslash S}\overline{\Phi}_{S}\overline{\Phi}_{K\backslash S}=\Phi_{J}\overline{\Phi}_{K}$.
\end{lem}
\begin{proof}
Writing $S:=J\cap K$, we face the integral
\begin{align*}
\nonumber \left(\overline{\Psi}_{I}\Psi_{J}\right)\star\left(\overline{\Psi}_{K}\Psi_{L}\right)&=\sigma_S\cdot\mathrm{e}^{-\left(\overline{\Psi},\Psi\right)}\overline{\Psi}_{I}\int\mathrm{d}(\overline{\Phi},\Phi)\,\Phi_{S}\Phi_{J\backslash S}\overline{\Phi}_{S}\overline{\Phi}_{K\backslash S}\\
   &\qquad \times\prod\limits_{\alpha\in M}
      \left(1+\overline{\phi}_\alpha\psi_\alpha+\overline{\psi}_\alpha\phi_\alpha
      -\overline{\phi}_\alpha\phi_\alpha-\overline{\phi}_\alpha\phi_\alpha\overline{\psi}_\alpha
	  \psi_\alpha\right)\Psi_{L},
\end{align*}
where we use
\begin{align*}
 \prod\limits_{\alpha\in M}
      \left(1+\overline{\phi}_\alpha\psi_\alpha+\overline{\psi}_\alpha\phi_\alpha
      -\overline{\phi}_\alpha\phi_\alpha-\overline{\phi}_\alpha\phi_\alpha\overline{\psi}_\alpha
	  \psi_\alpha\right)=
\mathrm{e}^{\left(\overline{\Psi},\Phi\right)}
      \mathrm{e}^{-\left(\overline{\Phi},\Phi\right)}
      \mathrm{e}^{\left(\overline{\Phi},\Psi \right)}
\end{align*}
as a consequence of (\ref{expfkt}).
In the next step we write $M=\left(M\backslash\left(J\cup K\right)\right)\dot{\cup} \left(J\backslash S\right)\dot{\cup} \left(K\backslash S\right) \dot{\cup} S$
(where $\dot{\cup}$ denotes a disjoint union) and arrive at
\begin{align*}
\nonumber\left(\overline{\Psi}_{I}\Psi_{J}\right)\star\left(\overline{\Psi}_{K}\Psi_{L}\right)  =&\sigma_S\sigma_{SJ}\cdot\mathrm{e}^{-\left(\overline{\Psi},\Psi\right)}\overline{\Psi}_{I}\int\mathrm{d}(\overline{\Phi},\Phi)\,\prod\limits_{\alpha\in S}\phi_\alpha\overline{\phi}_\alpha\\
\nonumber&\qquad\times    \prod\limits_{\alpha\in J\backslash S}\left( \phi_\alpha+\phi_\alpha\overline{\phi}_\alpha\psi_\alpha\right)
    \prod\limits_{\alpha\in K\backslash S}\left(\overline{\phi}_\alpha+\overline{\phi}_\alpha\overline{\psi}_\alpha\phi_\alpha\right)\\
   &\qquad \times\prod\limits_{\substack{\alpha\in M\\ \backslash\left(J\cup K\right)}}
      \left(1+\overline{\phi}_\alpha\psi_\alpha+\overline{\psi}_\alpha\phi_\alpha
      -\overline{\phi}_\alpha\phi_\alpha-\overline{\phi}_\alpha\phi_\alpha\overline{\psi}_\alpha
	  \psi_\alpha\right)\Psi_{L}.
\end{align*}
The sign $\sigma_{JS}:=\left(-1\right)^{|S|\left(|J\backslash S|+\frac{|S|-1}{2}\right)}$ occurs due to the permutation of all $\phi$'s in $\Phi_{S}$ with
all $\phi$'s in $\Phi_{J\backslash S}$, and $\Phi_{S}\overline{\Phi}_{S}=\left(-1\right)^{\frac{1}{2}|S|\left(|S|-1\right)}
  \left(\prod\limits_{\alpha\in S}\phi_\alpha\overline{\phi}_\alpha\right)$. Now we can perform the integration and arrive at
\begin{align*}
\nonumber& \left(\overline{\Psi}_{I}\Psi_{J}\right)\star\left(\overline{\Psi}_{K}\Psi_{L}\right)
 = \sigma_S\sigma_{JS}\cdot\mathrm{e}^{-\left(\overline{\Psi},\Psi\right)}
    \overline{\Psi}_{I}\prod\limits_{\alpha\in J\backslash S}\psi_\alpha
		\prod\limits_{\alpha\in K\backslash S}\overline{\psi}_\alpha
		\prod\limits_{\substack{\alpha\in M\\ \backslash\left(J\cup K\right)}}\left(1+\overline{\psi}_\alpha
	  \psi_\alpha\right)\Psi_{L},
\end{align*}
as claimed in (\ref{starprod}), since all involved sets are disjoint.
\end{proof}
There are several properties of the star product which follow from Lemma~\ref{basicstar}.
\begin{lem} \label{assoz}
 For all $\mu,\eta,\nu\in\cG_M$ we have
\begin{align*}
 \mu\star\left(\eta\star \nu\right)=\left(\mu\star \eta\right)\star \nu.
\end{align*}
\end{lem}
\begin{proof}
 By the definition of the star product we have
 \begin{align*}
\nonumber  \mu\star\left(\eta\star\nu\right)&=\mu\left(\overline{\psi},\psi\right)\star
      \int \mathrm{d}(\overline{\Phi},\Phi)\,\eta\left(\overline{\psi},\phi\right)\nu\left(\overline{\phi},\psi\right)
      \mathrm{e}^{-\left(\overline{\Psi},\Psi\right)+\left(\overline{\Psi},\Phi\right)
	-\left(\overline{\Phi},\Phi\right)+\left(\overline{\Phi},\Psi\right)}\\
\nonumber   &=\int\mathrm{d}(\overline{\Omega},\Omega)\int\mathrm{d}(\overline{\Phi},\Phi)\,\mu\left(\overline{\psi},\omega\right)
	\eta\left(\overline{\omega},\phi\right)\nu\left(\overline{\phi},\psi\right)\\
&\qquad\qquad\times\mathrm{e}^{-\left(\overline{\Psi},\Psi\right)+\left(\overline{\Psi},\Omega\right)
	-\left(\overline{\Omega},\Omega\right)+\left(\overline{\Omega},\Phi\right)
	-\left(\overline{\Phi},\Phi\right)+\left(\overline{\Phi},\Psi\right)}.
 \end{align*}
Performing the integration with respect to $\left(\overline{\phi},\phi\right)$ we gain
\begin{align*}
&\mu\star\left(\eta\star\nu\right)
=\int\mathrm{d}(\overline{\Omega},\Omega)\,\mu\left(\overline{\psi},\omega\right)
	\eta\left(\overline{\omega},\psi\right)\mathrm{e}^{-\left(\overline{\Psi},\Psi\right)+\left(\overline{\Psi},\Omega\right)
	-\left(\overline{\Omega},\Omega\right)+\left(\overline{\Omega},\Psi\right)}\star\nu\left(\overline{\psi},\psi\right),
\end{align*}
which is, in fact, $\left(\mu\star\eta\right)\star\nu$.
\end{proof}
According to the creation and annihilation operators on $\cB\left(\wedge\cH\right)$, there is also an implementation
of the CAR for the generators of $\cG_M$.
\begin{lem}\label{CARonG}
Let $\left\{\overline{\psi}_i,\psi_i\right\}_{i\in M}$ be the generators of $\cG_M$. For $\left\{\mu,\eta\right\}_{\star}:=\mu\star \eta+\eta\star \mu$ we have
\begin{align*}
 \left\{\psi_i,\psi_j\right\}_{\star}=\left\{\overline{\psi}_i,\overline{\psi}_j\right\}_{\star}=0,\qquad\mathrm{and}\qquad\left\{\overline{\psi}_i,\psi_j\right\}_{\star}=\delta_{ij}.
\end{align*}
 \end{lem}
\begin{proof}
 The identities follow directly from Lemma~\ref{basicstar} by an appropriate choice of $I,J,K$ and $L$.
We observe that
\begin{align*}
 \mathrm{e}^{-\left(\overline{\Psi},\Psi\right)}\prod\limits_{\substack{\alpha\in M\\ \backslash\left(J\cup K\right)}}\left(1+\overline{\psi}_\alpha
	  \psi_\alpha\right)=\prod\limits_{\substack{\alpha\in J\cup K}}\left(1-\overline{\psi}_\alpha
	  \psi_\alpha\right)
\end{align*}
and conclude for the first identity with $I=K=\emptyset$, and $J=\left\{i\right\}$, $L=\left\{j\right\}$ in (\ref{starprod})
that $S=\emptyset$ and, therefore, $\sigma_S=\sigma_{JS}=1$. This yields
\begin{align}
 \psi_i\star\psi_j=\left(1-\overline{\psi}_i\psi_i\right)\psi_i\psi_j=\psi_i\psi_j. \label{fz}
\end{align}
Setting $J=\left\{j\right\}$ and $L=\left\{i\right\}$, we gain $\psi_j\star\psi_i=\psi_j\psi_i$ and, hence,
$\psi_i\star\psi_j+\psi_j\star\psi_i=\psi_i\psi_j+\psi_j\psi_i=0$. Equivalently, we obtain $\overline{\psi}_i\overline{\psi}_j+\overline{\psi}_j\overline{\psi}_i=0$.\\
For the last identity we set $J=K=\emptyset$, $I=\left\{i\right\}$ and $L=\left\{j\right\}$. On the one hand, (\ref{starprod}) leads to
\begin{align*}
\overline{\psi}_i\star\psi_j=\overline{\psi}_i\psi_j,
\end{align*}
which is valid for both $i=j$ and $i\neq j$. 
On the other hand, with $I=L=\emptyset$, and $J=\left\{j\right\}$ and $K=\left\{i\right\}$, we have to distinguish between the cases $J=K$ and $J\neq K$.
For $J\neq K$ we have
\begin{align*}
\nonumber \psi_j\star\overline{\psi}_i&=\left(1-\overline{\psi}_i\psi_i\right)\left(1-\overline{\psi}_j\psi_j\right)\psi_j\overline{\psi}_i\\
      &=\psi_j\overline{\psi}_i.
\end{align*}
For $J=K$ we have $S=J=K$ and thus
\begin{align}
\psi_j\star\overline{\psi}_i=\left(1-\overline{\psi}_i\psi_i\right). \label{sz}
\end{align}
Together, the last two results give $\psi_j\star\overline{\psi}_i=\delta_{ij}-\overline{\psi}_i\psi_j$. Finally, we arrive at
$\overline{\psi}_i\star\psi_j+\psi_j\star\overline{\psi}_i=\delta_{ij}$. We mention that in (\ref{fz})-(\ref{sz}) $\sigma_S=\sigma_{JS}=1$
due to the choice of the sets $I,\ J,\ K$ and $L$.
\end{proof}
By a straightforward calculation using Lemma~\ref{basicstar} one can also show that for any generator $\left\{\overline{\psi}_i,\psi_i\right\}_{i\in M}$ of $\cG_M$ we have
the following:
\begin{cor} \label{sterneigen}
Let $\left\{\overline{\psi}_i,\psi_i\right\}_{i\in M}$ be the generators of $\cG_M$. Then we have
\begin{align*}
 \overline{\psi}_{i_1}\star\cdots\star\overline{\psi}_{i_m}\star\psi_{j_1}\star\cdots\star\psi_{j_n}
    =\overline{\psi}_{i_1}\cdots\overline{\psi}_{i_m}\psi_{j_1}\cdots\psi_{j_n}.
\end{align*}
\end{cor}
\begin{proof}
 We use the associativity $\overline{\psi}_{i_1}\star\cdots\star\overline{\psi}_{i_m}\star\psi_{j_1}\star\cdots\star\psi_{j_n}
=\left(\overline{\psi}_{i_1}\star\cdots\star\overline{\psi}_{i_m}\right)\star\left(\psi_{j_1}\star\cdots\star\psi_{j_n}\right)$
and calculate the brackets using Lemma~\ref{basicstar}. For the first bracket we set in (\ref{starprod}) $I=\left\{i_1,\dots,i_m\right\}$ and
$J=K=L=\emptyset$. For the second bracket we use $I=J=K=\emptyset$ and $L=\left\{j_1,\dots,j_n\right\}$. For both we 
have $\sigma_S=\sigma_{JS}=1$ and we conclude 
 \begin{align*}
 \overline{\psi}_{i_1}\star\cdots\star\overline{\psi}_{i_m}\star\psi_{j_1}\star\cdots\star\psi_{j_n}
    =\left(\overline{\psi}_{i_1}\cdots\overline{\psi}_{i_m}\right)\star\left(\psi_{j_1}\cdots\psi_{j_n}\right).
\end{align*}
The last star product can be calculated by setting $I=\left\{i_1,\dots,i_m\right\}$, $L=\left\{j_1,\dots,j_n\right\}$ and $J=K=\emptyset$ in (\ref{starprod}).
Again, $\sigma_S=\sigma_{JS}=1$ and we arrive at the assertion.
\end{proof}
We emphasize that
\begin{align*}
 \overline{\psi}_i\psi_j=\overline{\psi}_i\star\psi_j,\quad\text{but}\quad\psi_i\overline{\psi}_j=-\overline{\psi}_j\psi_i=
	-\overline{\psi}_j\star\psi_i.
\end{align*}
This implies that the star product can be inserted (or skipped) only if the monomial in $\psi$ and $\overline{\psi}$ is 
normal-ordered (i.e., all $\overline{\psi}$'s are to the left of all $\psi$'s). As follows from the proof, monomials containing only $\psi$'s or $\overline{\psi}$'s can also be considered as normal-ordered in the sense that
we can can identify
$\overline{\psi}_{i_1}\star\cdots\star\overline{\psi}_{i_m}=\overline{\psi}_{i_1}\cdots\overline{\psi}_{i_m}$ and 
$\psi_{j_1}\star\cdots\star\psi_{j_n}=\psi_{j_1}\cdots\psi_{j_n}$.
\begin{lem} \label{prodrule}
   Let $N\in\mathbbm{N}$ and $A_i\in\mathcal{B}\left(\wedge\cH\right)$ for $i\in\left\{1,\dots,N\right\}$. Then
\begin{align}
 \Theta\left(A_1 A_2\cdots A_N\right)=\Theta\left(A_1\right)
      \star\Theta\left(A_2\right)\star\dots\star\Theta\left(A_N\right). \label{starvswedge}
\end{align}
\end{lem}
\begin{proof}
Due to the associativity of the star product it suffices to consider the assertion for $N=2$. We use the CAR 
to establish normal order in the product $A_1 A_2 \in\cB\left(\wedge\cH\right)$ and indicate this order
by $\substack{\bullet\\ \bullet}\, A_1 A_2\, \substack{\bullet\\ \bullet}$. For some $a_{\substack{i_1\dots i_m\\ j_1\dots j_n}}\in\mathbbm{C}$, we can write 
\begin{align*}
  \substack{\bullet\\ \bullet}\,A_1 A_2\,\substack{\bullet\\ \bullet}\ =\sum\limits_{m,n}\sum\limits_{\substack{i_1\dots i_m\\ j_1\dots j_n}\in M}a_{\substack{i_1\dots i_m\\ j_1\dots j_n}}
	\e_{i_1}\cdots\e_{i_m}\v_{j_1}\cdots\v_{j_n}
\end{align*}
and apply $\Theta$. Together with Corollary \ref{sterneigen} we arrive at
\begin{align}
 \Theta\left(\substack{\bullet\\ \bullet}\,A_1A_2\,\substack{\bullet\\ \bullet}\right)=\sum\limits_{m,n}\sum\limits_{\substack{i_1\dots i_m\\ j_1\dots j_n}\in M}a_{\substack{i_1\dots i_m\\ j_1\dots j_n}}\,
      \overline{\psi}_{i_1}\star\cdots\star\overline{\psi}_{i_m}\star\psi_{j_1}\star\cdots\star\psi_{j_m}. \label{the12}
\end{align}
Now we can use the CAR on $\cG_M$ to restore the same order we had in $A_1 A_2$ within the r.h.s.\ of (\ref{the12}) and recognize that it 
equals $\Theta\left(A_1\right)\star\Theta\left(A_2\right)$. In other words, we have
\begin{align*}
 \sum\limits_{m,n}\sum\limits_{\substack{i_1\dots i_m\\ j_1\dots j_n}\in M}a_{\substack{i_1\dots i_m\\ j_1\dots j_n}}\,
    \overline{\psi}_{i_1}\star\cdots\star\overline{\psi}_{i_m}\star\psi_{j_1}\star\cdots\star\psi_{j_m}
     =\ \substack{\bullet\\ \bullet}\,\Theta\left(A_1\right)\star\Theta\left(A_2\right)\,\substack{\bullet\\ \bullet},
\end{align*}
which gives the assertion. 
\end{proof} 
We can equip $\left(\cG_M, +, \star\right)$ with an involution $\left(\,\cdot\,\right)^*$ such
that $\left(\cG_M, +, \star\right)$ becomes a *-algebra. 
\begin{defn} \label{invol}
 For all $\mu_i\in\cG_M$, $i\in\mathbbm{N}$, and $c\in\mathbbm{C}$, the involution $\left(\,\cdot\,\right)^*$ on $\left(\cG_M, +, \star\right)$ is defined by
 $\left(\psi_i\right)^*:=\overline{\psi}_i$ and $\left(\overline{\psi}_i\right)^*:=\psi_i$ $\forall i\in M$, and 
 \begin{align*}
  \left(c\,\mu_1\cdots \mu_n\right)^*:=\overline{c}\,\mu_n^*\cdots \mu_1^*.
 \end{align*}
\end{defn}
\begin{rem} \label{invorul}
 For $\mu\equiv\mu\left(\overline{\psi},\phi\right):=\sum\limits_{I,J}a_{IJ}\,\overline{\Psi}_I\Phi_J$ and $a_{IJ}\in\mathbbm{C}$, the involution
 $\mu^*$ is given by $\mu^*\left(\overline{\phi},\psi\right)=\sum\limits_{I,J}\overline{a}_{IJ}\,\overline{\Phi}_{J'}\Psi_{I'}=
  \sum\limits_{I,J}\left(-1\right)^{\frac{1}{2}|I|\left(|I|-1\right)+\frac{1}{2}|J|\left(|J|-1\right)}
      \overline{\alpha}_{IJ}\,\overline{\Phi}_J\Psi_I$.
 We emphasize that $\left(\mu\left(\overline{\psi},\phi\right)\right)^*=\mu^*\left(\overline{\phi},\psi\right)\neq \left(\mu\left(\overline{\phi},\psi\right)\right)^*$.
\end{rem}
\begin{lem}
The involution in Definition \ref{invol} is compatible with $\Theta$, the Grassmann integration, and the star product:
\begin{enumerate}
 \item[a)] $\Theta\left(\left(\,\cdot\,\right)^*\right)=\left(\Theta\left(\,\cdot\,\right)\right)^*$,
 \item[b)] $\int\mathrm{d}(\overline{\Psi},\Psi)\left(\,\cdot\,\right)^*=
	      \left[\int\mathrm{d}(\overline{\Psi},\Psi)\left(\,\cdot\,\right)\right]^*$,
 \item[c)] $\left(\mu\star \eta\right)^*=\eta^*\star \mu^*$.
\end{enumerate}
\end{lem}
\begin{proof}
 We prove $a)$ and $b)$. $c)$ is a consequence of $b)$.
 \begin{enumerate}
  \item[a)] For any $I,J\subseteq M$, we abbreviate 
	  $C^*_I:=\e_{i_1}\cdots\e_{i_m}$ and $C_J:=\v_{j_1}\cdots\v_{j_n}$ and write any $A\in\cB\left(\cH\right)$ as
	  $A=\sum\limits_{I,J}a_{IJ}\,C^*_I C_J$ for some $a_{IJ}\in\mathbbm{C}$. This leads to
	  \begin{align*}
	  \nonumber \left(\Theta\left(A\right)\right)^*&=\bigg(\sum\limits_{I,J}a_{IJ}\,\overline{\Psi}_I\Psi_J\bigg)^*
	    =\sum\limits_{I,J}\overline{a}_{IJ}\,\overline{\Psi}_{J'}\Psi_{I'}=\Theta\bigg(\sum\limits_{I,J}\overline{a}_{IJ}\,C^*_{J'}C_{I'}\bigg)\\
	    &=\Theta\bigg(\bigg(\sum\limits_{I,J}a_{IJ}\,C^*_I C_J\bigg)^*\bigg)=\Theta\left(A^*\right).
	  \end{align*}
  \item[b)] For a fixed, but arbitrary $i\in M$ and $\mu\in\cG_M$ we formally have
	    $\left(\frac{\delta}{\delta\overline{\psi}_i}\frac{\delta}{\delta\psi_i}\right)^*\mu=\frac{\delta}{\delta\overline{\psi}_i}\frac{\delta}{\delta\psi_i}\mu$,
	    which gives the assertion.
  \item[c)] We calculate the l.h.s. of $c)$ according to $b)$ and Remark \ref{invorul}:
	\begin{align*}
\nonumber	 \left(\mu\star\eta\right)^*&=
	  \int\mathrm{d}(\overline{\Phi},\Phi)\,
      \eta^*\left(\overline{\psi},\phi\right)\mu^*\left(\overline{\phi},\psi\right)
      \mathrm{e}^{-\left(\overline{\Psi},\Psi\right)}
      \mathrm{e}^{\left(\overline{\Psi},\Phi\right)}
      \mathrm{e}^{-\left(\overline{\Phi},\Phi\right)}
      \mathrm{e}^{\left(\overline{\Phi},\Psi \right)}\\
      &=\eta^*\star\mu^*,
	\end{align*}
    since $\left(\mathrm{e}^{\left(\,\cdot\,\right)}\right)^*=\mathrm{e}^{\left(\,\cdot\,\right)}$.
 \end{enumerate}
\end{proof}
A key property of the Grassmann integral for deriving representability conditions as in the next section is the cyclicity property which has its equivalent in the cyclicity of the trace, i.e.,
$\mathrm{tr}\left\{AB\right\}=\mathrm{tr}\left\{BA\right\}$.
\begin{thm}  \label{cyclic1}
 For $\mu,\eta\in\cG_M$, we have
\begin{align*}
 \int\cD(\overline{\Psi},\Psi)\left(\mu\star\eta\right)  
   =\int\cD(\overline{\Psi},\Psi)\left(\eta\star\mu\right)  .
\end{align*}
\end{thm}
\begin{proof}
Without loss of generality, we can set  
\begin{align*}
 \mu:=\overline{\Psi}_I\Psi_J\quad\text{and}\quad \eta:=\overline{\Psi}_K\Psi_L
\end{align*}
and observe with (\ref{starprod}) and $T:=I\cap L$ 
\begin{align*}
 &\int\cD(\overline{\Psi},\Psi)\,\mu\star\eta=\sigma_S\sigma_T \sigma_{JS}\int\cD(\overline{\Psi},\Psi)
    \cdot\mathrm{e}^{-\left(\overline{\Psi},\Psi\right)}\\
    &\qquad\qquad\qquad\qquad\qquad\qquad\times\overline{\Psi}_T\overline{\Psi}_{I\backslash T}\prod\limits_{\alpha\in J\backslash S}\psi_\alpha
		\prod\limits_{\alpha\in K\backslash S}\overline{\psi}_\alpha
		\prod\limits_{\substack{\alpha\in M\\ \backslash\left(J\cup K\right)}}\left(1+\overline{\psi}_\alpha
	  \psi_\alpha\right)\Psi_T\Psi_{L\backslash T}.
\end{align*}
Afterwards, we rearrange the factors and arrive at
\begin{align}
\nonumber &\int\cD(\overline{\Psi},\Psi)\,\mu\star\eta=\sigma_S\sigma_T\tilde{\sigma}\int\mathrm{d}(\overline{\Psi},\Psi)\,
        \overline{\Psi}_{I\backslash T}\overline{\Psi}_{K\backslash S}\Psi_{J\backslash S}\Psi_{L\backslash T}\prod\limits_{\alpha\in T}\overline{\psi}_\alpha\psi_\alpha\\
  &\qquad\qquad\qquad\qquad\qquad\qquad\times  \prod\limits_{\alpha\in M}\left(1+\overline{\psi}_\alpha\psi_\alpha\right) 
 \prod\limits_{\substack{\alpha\in M\\ \backslash\left(J\cup K\right)}}\left(1+\overline{\psi}_\alpha\psi_\alpha\right) ,    \label{Int2}
\end{align}
where $\tilde{\sigma}\in\left\{\pm 1\right\}$ corresponds to the signs resulting from the anticommutations and is given by 
\begin{align*}
           \tilde{\sigma}:=\left(-1\right)^{|S||J\backslash S|+|T||K\backslash S|+\frac{1}{2}|S|\left(|S|-1\right)
      +\frac{1}{2}|T|\left(|T|-1\right)+|T||J\backslash S|+|T||I\backslash T|+|K\backslash S||J\backslash S|}.                 
\end{align*}
To go on, we need some preparation. First of all, we observe that
\begin{align*}
\nonumber &\prod\limits_{\alpha\in M}\left(1+\overline{\psi}_\alpha\psi_\alpha\right)
\prod\limits_{\substack{\alpha\in M\\ \backslash\left(J\cup K\right)}}\left(1+\overline{\psi}_\alpha\psi_\alpha\right)
=\prod\limits_{\substack{\alpha\in M\\ \backslash\left(J\cup K\right)}}\left(1+2\overline{\psi}_\alpha\psi_\alpha\right)
  \prod\limits_{\alpha\in J\cup K}\left(1+\overline{\psi}_\alpha\psi_\alpha\right).
\end{align*}
On the one hand, we have $J\cup K=\left(J\backslash S\right)\dot{\cup}\left(K\backslash S\right)\dot{\cup} S$, which implies:
\begin{align*}
   \prod\limits_{\alpha\in J\cup K}\left(1+\overline{\psi}_\alpha\psi_\alpha\right) \overline{\Psi}_{K\backslash S}\Psi_{J\backslash S}
  =  \prod\limits_{\alpha\in S}\left(1+\overline{\psi}_\alpha\psi_\alpha\right) \overline{\Psi}_{K\backslash S}\Psi_{J\backslash S}.
\end{align*}
On the other hand, we have by the same arguments
\begin{align*}
\nonumber & \prod\limits_{\substack{\alpha\in M\\ \backslash\left(J\cup K\right)}}\left(1+2\overline{\psi}_\alpha\psi_\alpha\right) 
    \overline{\Psi}_{I\backslash T}\overline{\Psi}_{K\backslash S}\Psi_{J\backslash S}\Psi_{L\backslash T}\prod\limits_{\alpha\in T}\overline{\psi}_\alpha\psi_\alpha\\
  &\qquad\qquad\qquad= \prod\limits_{\substack{\alpha\in M\\ \backslash\left(J\cup K\cup I\cup L\right)}}\left(1+2\overline{\psi}_\alpha\psi_\alpha\right) 
    \overline{\Psi}_{I\backslash T}\overline{\Psi}_{K\backslash S}\Psi_{J\backslash S}\Psi_{L\backslash T}\prod\limits_{\alpha\in T}\overline{\psi}_\alpha\psi_\alpha,
\end{align*}
since $I\cup L\equiv\left(I\backslash T\right)\dot{\cup} \left(L\backslash T\right)\dot{\cup} T$. Consequently, our latter calculations
lead in~(\ref{Int2}) to 
\begin{align}
 \nonumber &\int\cD(\overline{\Psi},\Psi)\,\mu\star\eta=\sigma_S\sigma_T\tilde{\sigma}
  \int\mathrm{d}(\overline{\Psi},\Psi)\,
        \overline{\Psi}_{I\backslash T}\overline{\Psi}_{K\backslash S}\Psi_{J\backslash S}\Psi_{L\backslash T}\prod\limits_{\alpha\in T}\overline{\psi}_\alpha\psi_\alpha\\
  &\qquad\qquad\qquad\qquad\qquad\times  \prod\limits_{\alpha\in S}\left(1+\overline{\psi}_\alpha\psi_\alpha\right) 
 \prod\limits_{\substack{\alpha\in M\\ \backslash\left(J\cup K\cup I\cup L\right)}}\left(1+2\overline{\psi}_\alpha\psi_\alpha\right) . \label{laseq}
\end{align}
Let us take a closer look at the involved sets. First of all, we observe
\begin{enumerate}
 \item[(I)] $K\backslash S \cap J\backslash S=\emptyset$
 \item[(II)] $I\cup \left(K\backslash S\right) = L\cup \left(J\backslash S\right)$
 \item[(III)] $I\cap \left(K\backslash S\right) = \emptyset$
 \item[(IV)] $L\cap \left(J\backslash S\right) = \emptyset$.
\end{enumerate}
In any other case we have $\int\cD(\overline{\Psi},\Psi)\,\mu\star\eta=\int\cD(\overline{\Psi},\Psi)\,\eta\star\mu=0$. 
These observations have some consequences:
\begin{enumerate}
 \item[a)] (II) and (I) $\Rightarrow$ $\left(K\backslash S\right)\subseteq L$ and $\left(J\backslash S\right)\subseteq I$
  $\Rightarrow$ $\exists$ $T_1,T_2\subseteq M$ s.th. $I=\left(J\backslash S\right)\dot{\cup} T_1$ and $L=\left(K\backslash S\right)\dot{\cup} T_2$.
 \item[b)] (III) and $I=\left(J\backslash S\right)\dot{\cup} T_1$ $\Rightarrow$ $\left(\left(J\backslash S\right)\dot{\cup} T_1\right)\cap \left(K\backslash S\right)=\emptyset$
	    $\Rightarrow$ $T_1\cap K\backslash S=\emptyset$. Analogously: (IV) and $L=\left(K\backslash S\right)\dot{\cup} T_1$ $\Rightarrow$ 
	    $T_2\cap\left(J\backslash S\right)=\emptyset$.
 \item[c)] (II) and b) $\Rightarrow$ $T_1=T_2$, since all sets on the l.h.s. and r.h.s. of (II) are disjoint, respectively.
 \item[d)] a), b) and c) $\Rightarrow$ $L\cap I=\left(\left(K\backslash S\right)\dot{\cup} T_1\right)\cap \left(\left(J\backslash S\right)\dot{\cup} T_2\right)=T_1\cap T_2=:T$.
\end{enumerate}
Back to a), we see that $I=\left(J\backslash S\right)\dot{\cup} T$ or $I\backslash T=J\backslash S$, and that 
    $L=\left(K\backslash S\right)\dot{\cup} T$ implies $L\backslash T=K\backslash S$. This is illustrated in the 
following figure.
\begin{figure}[h]
    \centering
    \includegraphics[scale=0.6]{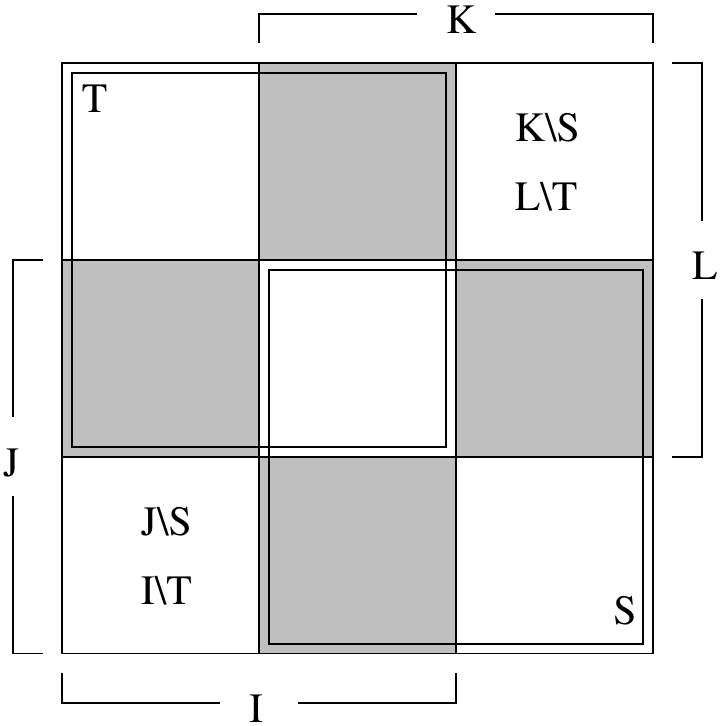}
    \captionsetup{margin=23pt,font=small,labelfont=bf,format=hang}
    \caption*{Breteaux chequerboard: {\emph{The integrals vanish if $J\cup L \neq I\cup K$. $S:=J\cap K$ and $T:=I\cap L$. Grey areas represent empty subsets.}}}
\end{figure}\\
We go on in (\ref{laseq}) and take the intersection $S\cap T$ into account. The term
$\prod\limits_{\alpha\in T}\overline{\psi}_\alpha\psi_\alpha \prod\limits_{\alpha\in S}\left(1+\overline{\psi}_\alpha\psi_\alpha\right) $
contributes to the the integral as follows:
\begin{align*}
\nonumber& \prod\limits_{\alpha \in T\cup S}\frac{\delta}{\delta\overline{\psi}_\alpha}\frac{\delta}{\delta\psi_\alpha} 
  \prod\limits_{\alpha\in T}\overline{\psi}_\alpha\psi_\alpha\prod\limits_{\beta\in S}\left(1+\overline{\psi}_\beta\psi_\beta\right)=\prod\limits_{\alpha \in T\cup S}\frac{\delta}{\delta\overline{\psi}_\alpha}\frac{\delta}{\delta\psi_\alpha}
  \prod\limits_{\alpha\in T\cup S}\overline{\psi}_\alpha\psi_\alpha,
\end{align*}
since $\prod\limits_{\alpha\in T\cap S}\overline{\psi}_\alpha\psi_\alpha\prod\limits_{\beta\in T\cap S}\left(1+\overline{\psi}_\beta\psi_\beta\right)
  =\prod\limits_{\alpha\in T\cap S}\overline{\psi}_\alpha\psi_\alpha$ and
\begin{align*}
 &\prod\limits_{\alpha \in S\backslash \left(T\cap S\right)}\frac{\delta}{\delta\overline{\psi}_\alpha}\frac{\delta}{\delta\psi_\alpha}
	\prod\limits_{\beta\in S\backslash \left(T\cap S\right)}\left(1+\overline{\psi}_\beta\psi_\beta\right)
=\prod\limits_{\alpha \in S\backslash \left(T\cap S\right)}\frac{\delta}{\delta\overline{\psi}_\alpha}\frac{\delta}{\delta\psi_\alpha}
	\prod\limits_{\beta\in S\backslash \left(T\cap S\right)}\overline{\psi}_\beta\psi_\beta.
\end{align*}
This finishes our calculations and we conclude:
\begin{align}
 \nonumber &\int\cD(\overline{\Psi},\Psi)\,\mu\star\eta=\sigma_S\sigma_T\tilde{\sigma}
    \int\mathrm{d}(\overline{\Psi},\Psi)\,
        \prod\limits_{\alpha\in T\cup S}\overline{\psi}_\alpha\psi_\alpha\\
  &\qquad\qquad\qquad\qquad\qquad\times 
 \prod\limits_{\substack{\alpha\in M\\ \backslash\left(J\cup K\cup I\cup L\right)}}\left(1+2\overline{\psi}_\alpha\psi_\alpha\right) 
      \overline{\Psi}_{I\backslash T}\overline{\Psi}_{K\backslash S}\Psi_{J\backslash S}\Psi_{L\backslash T}.    \label{endcalc}
\end{align}
The r.h.s.\ of the assertion in Theorem~\ref{cyclic1} can be calculated analogously. The result is
\begin{align*}
 \nonumber &\int\cD(\overline{\Psi},\Psi)\,\eta\star\mu=\sigma_T\sigma_S\widehat{\sigma}\int\mathrm{d}(\overline{\Psi},\Psi)\,
        \prod\limits_{\alpha\in S\cup T}\overline{\psi}_\alpha\psi_\alpha\\
  &\qquad\qquad\qquad\qquad\qquad\times 
 \prod\limits_{\substack{\alpha\in M\\ \backslash\left(J\cup K\cup I\cup L\right)}}\left(1+2\overline{\psi}_\alpha\psi_\alpha\right) 
      \overline{\Psi}_{K\backslash S}\overline{\Psi}_{I\backslash T}\Psi_{L\backslash T}\Psi_{J\backslash S},    
\end{align*}
where the sign resulting from the anticommutations is
\begin{align*}
           \widehat{\sigma}:=\left(-1\right)^{|T||L\backslash T|+|S||L\backslash T|+\frac{1}{2}|S|\left(|S|-1\right)
      +\frac{1}{2}|T|\left(|T|-1\right)+|S||I\backslash T|+|S||K\backslash S|+|I\backslash T||L\backslash T|}.                 
\end{align*}
The l.h.s. and the r.h.s. of Theorem~\ref{cyclic1} are symmetric with respect to the involved sets. The proof is complete by the observation
\begin{align*}
 \tilde{\sigma}=\widehat{\sigma}=\left(-1\right)^{\frac{1}{2}|S|\left(|S|-1\right)
      +\frac{1}{2}|T|\left(|T|-1\right)+|K\backslash S||J\backslash S|+|T||K\backslash S|+|S||J\backslash S|},
\end{align*}
which follows from $I\backslash T=J\backslash S$ and $L\backslash T=K\backslash S$.
\end{proof}
\begin{rem}\label{intlsg}
 The integral on the r.h.s.\ of (\ref{endcalc}) can be carried out. Abbreviating $s_Q:=\frac{1}{2}|Q|\left(|Q|-1\right)$ for $Q\subseteq M$, we have
\begin{align*}
\nonumber \int\cD(\overline{\Psi},\Psi)\,\mu\star\eta&=\sigma_S\sigma_T\left(-1\right)^{s_S+s_T+|T||K\backslash S|+|S||J\backslash S|+s_{I\backslash T}+s_{K\backslash S}}\\
\nonumber&\qquad\times\int\mathrm{d}(\overline{\Psi},\Psi)\,\prod\limits_{\alpha\in I\backslash T}\overline{\psi}_\alpha\psi_\alpha\prod\limits_{\alpha\in K\backslash S}\overline{\psi}_\alpha\psi_\alpha\prod\limits_{\alpha\in T\cup S}\overline{\psi}_\alpha\psi_\alpha
\prod\limits_{\alpha\in M\backslash \left(I\cup K\right)}\left(1+2\overline{\psi}_\alpha\psi_\alpha\right)\\
\nonumber &=\sigma_S\sigma_T\left(-1\right)^{s_S+s_T+|T||K\backslash S|+|S||J\backslash S|+s_{I\backslash T}+s_{K\backslash S}}\\
\nonumber &\qquad\times\left(-1\right)^{|I\backslash T|+|K\backslash S|+|T\cup S|}\left(-2\right)^{|M|-|I\cup K|}.
\end{align*}
With $|I\backslash T|+|K\backslash S|+|T\cup S|=|I\cup K|$ we obtain
\begin{align*}
 \int\cD(\overline{\Psi},\Psi)\,\mu\star\eta=\sigma_S\sigma_T\left(-1\right)^{s_J+s_L}2^{|M|-|I\cup K|}
\end{align*}
for $\mu:=\overline{\Psi}_I\Psi_J$ and $\eta:=\overline{\Psi}_K\Psi_L$.
\end{rem}
\begin{rem}
A consequence of Lemma~\ref{assoz} and \ref{cyclic1} is the invariance of the Grassmann integral with respect to cyclic permutations of the integrand:
\begin{align}
\int\mathrm{d}(\overline{\Psi},\Psi)\left(\mu_1\star\mu_2\star\cdots\star\mu_N\right)
	      =\int\mathrm{d}(\overline{\Psi},\Psi)\left(\mu_2\star\cdots\star\mu_N\star\mu_1\right).
	  \label{cyclic2}
\end{align}
This also holds true for $\int\cD(\overline{\Psi},\Psi)\left(\,\cdot\,\right)$, since $\mathrm{e}^{2\left(\overline{\Psi},\Psi\right)}$
commutes with any $\mu\in\cG_M$.
\end{rem}

Given an involution on $\left(\cG_M, +, \star\right)$, we define the 
property of positivity on $\cG_M$ as follows.
\begin{defn}
 We call $\mu\in\cG_M$ positive semi-definite, shortly $\mu\geq 0$, if there exists
an $\eta \in\cG_M$ such that
\begin{align*}
\mu= \eta^*\star\eta.
\end{align*}
\end{defn}

Approaching the problem of representability by Grassmann integration, an important result is the following theorem.
\begin{thm} \label{positivity}
 For any $\mu\in\cG_M$ with $\mu\geq 0$ we have
\begin{align} 
 \left(-1\right)^{|M|}\int\cD(\overline{\Psi},\Psi)\, \mu  \geq 0 . \label{posit}
\end{align}
\end{thm}
\begin{proof}
We use an induction in $|M|\in\mathbbm{N}$. For this purpose, we write any $\xi\in\cG_{M+1}:=
\mathrm{span}\left\{\overline{\psi}_1,\dots,\overline{\psi}_{|M|},\overline{\psi}_{|M|+1},\psi_1,\dots,\psi_{|M|},\psi_{|M|+1}\right\}$
as 
\begin{align}
 \xi=\eta_{00}+\eta_{01}\psi_{|M|+1}+\overline{\psi}_{|M|+1}\eta_{10}+\overline{\psi}_{|M|+1}\eta_{11}\psi_{|M|+1} \label{deco}
\end{align}
 for normal-ordered $\eta_{00},\, \eta_{01},\, \eta_{10},\, \eta_{11}\in\cG_M$. We indicate integration with respect
to a certain index set $M$ by writing $\int\mathrm{d}_{M}(\overline{\Psi},\Psi)$ and $\int\cD_{M}(\overline{\Psi},\Psi)$, respectively.
Furthermore, we recall that
\begin{align*}
\nonumber \mathrm{e}^{E_M}&:=\mathrm{e}^{\left(\overline{\Psi},\Psi\right)}
      \mathrm{e}^{\left(\overline{\Psi},\Phi\right)}
      \mathrm{e}^{-\left(\overline{\Phi},\Phi\right)}
      \mathrm{e}^{\left(\overline{\Phi},\Psi \right)}\\
      &=\prod\limits_{\alpha=1}^{M}\left(1-\overline{\phi}_\alpha \phi_\alpha+\overline{\psi}_\alpha\psi_\alpha
	  +\overline{\phi}_\alpha\psi_\alpha+\overline{\psi}_\alpha\phi_\alpha-2\overline{\psi}_\alpha\psi_\alpha\overline{\phi}_\alpha\phi_\alpha\right).
\end{align*}
In order to show (\ref{posit}) for $|M|=0$, we consider $\mu:=a^*\star a\in\cG_0$ with $a\in\mathbbm{C}$, and observe that with $\int\cD_0(\overline{\Psi},\Psi)=1$ the l.h.s.\ of (\ref{posit}) is nonnegative,
\begin{align*}
 \int\cD_0(\overline{\Psi},\Psi)\,\mu = \left|a\right|^2\geq 0.
\end{align*}
Now we assume that (\ref{posit}) holds for $|M|$ and consider the l.h.s.\ of (\ref{posit}) for $|M|+1$ and 
$\mu=\xi^*\star\xi$. We abbreviate $\psi_{|M|+1}\equiv \psi'$ and $\overline{\psi}_{|M|+1}\equiv\overline{\psi'}$. 
\begin{align}
\nonumber &\left(-1\right)^{|M|+1}\int\cD_{M+1}(\overline{\Psi},\Psi)\left(\xi^*\star\xi\right)\\
\nonumber  &\qquad=\left(-1\right)^{|M|+1}\int\cD_{M+1}(\overline{\Psi},\Psi)\Big[\eta_{00}^*\star\eta_{00}+\eta_{00}^*\star\left(\overline{\psi'}\,\eta_{11}\psi'\right)+\left(\overline{\psi'}\,\eta_{01}^*\right)\star\left(\eta_{01}\psi'\right)\big.\\
  &\phantom{....................}+\left(\eta_{10}^*\psi'\right)\star\left(\overline{\psi'}\,\eta_{10}\right)+\left(\overline{\psi'}\,\eta_{11}^*\psi'\,\right)\star\eta_{00}\big.
  +\left(\overline{\psi'}\,\eta_{11}^*\psi'\right)\star\left(\overline{\psi'}\,\eta_{11}\psi'\right)\Big].  \label{grogla}
\end{align}
Other terms like $\int\cD_{M+1}(\overline{\Psi},\Psi)\eta_{00}^*\star\left(\eta_{01}\psi'\right)$ vanish, as can 
be seen in (\ref{laseq}), since, in this case, $I\cup K\neq J\cup L$.

In the next step, we use the definition of the star product and the identity $\int\mathrm{d}_{M+1}(\overline{\Psi},\Psi)
    =\int\mathrm{d}_{M}(\overline{\Psi},\Psi)\frac{\delta}{\delta\overline{\psi'}}\frac{\delta}{\delta\psi'}$ 
to carry out all integrations with respect to $\psi'$ and $\overline{\psi'}$. We exemplify this step by
the last term on the r.h.s.\ of (\ref{grogla}):
\begin{align*}
\nonumber &\left(-1\right)^{|M|+1}\int\cD_{M+1}(\overline{\Psi},\Psi)\left(\overline{\psi'}\,\eta_{11}^*\psi'\right)\star\left(\overline{\psi'}\,\eta_{11}\psi'\right)\\
\nonumber &\qquad = \left(-1\right)^{|M|+1}\int\mathrm{d}_{M+1}(\overline{\Psi},\Psi)\int\mathrm{d}_{M+1}\left(\overline{\Phi},\Phi\right)
	\overline{\psi'}\,\eta_{11}^*\left(\overline{\psi},\phi\right)\phi'
\overline{\phi'}\,\eta_{11}\left(\overline{\phi},\psi\right)\psi'\,\mathrm{e}^{E_{M+1}}.
\end{align*}
Since $\eta_{11}^*\left(\overline{\psi},\phi\right)\eta_{11}\left(\overline{\phi},\psi\right)$ is even in 
the $\left(\overline{\psi},\psi,\overline{\phi},\phi\right)$ variables, we continue with
\begin{align*}
\nonumber &\left(-1\right)^{|M|+1}\int\cD_{M+1}(\overline{\Psi},\Psi)\left(\overline{\psi'}\,\eta_{11}^*\psi'\right)\star\left(\overline{\psi'}\,\eta_{11}\psi'\right)\\
\nonumber &\qquad = \left(-1\right)^{|M|+1} \int\mathrm{d}_{M}(\overline{\Psi},\Psi)\int\mathrm{d}_{M}\left(\overline{\Phi},\Phi\right)\eta_{11}^*\left(\overline{\psi},\phi\right)\eta_{11}\left(\overline{\phi},\psi\right)\mathrm{e}^{E_M}\\
\nonumber &\qquad\quad \times \frac{\delta}{\delta\overline{\phi'}}\frac{\delta}{\delta\phi'}\frac{\delta}{\delta\overline{\psi'}}\frac{\delta}{\delta\psi'}\overline{\psi'}\phi'\overline{\phi'}\psi'\left(1-\overline{\phi'}\phi'+\overline{\psi'}\psi'+\overline{\phi'}\psi'+\overline{\psi'}\phi'-2\overline{\psi'}\psi'\overline{\phi'}\phi'\right)\\
	  &\qquad = \left(-1\right)^{|M|+2}\int\cD_{M}(\overline{\Psi},\Psi)\,\eta_{11}^*\star\eta_{11}.
\end{align*}
By analogous calculations, we obtain
\begin{align*}
 \nonumber &\left(-1\right)^{|M|+1}\int\cD_{M+1}(\overline{\Psi},\Psi)\left(\xi^*\star\xi\right)\\
 \nonumber &\qquad=\left(-1\right)^{|M|+2}\int\cD_{M}(\overline{\Psi},\Psi)\Big[2\eta_{00}^*\star\eta_{00}+\eta_{00}^*\star\tilde{\eta}_{11}
		+\eta_{01}^*\star\eta_{01}+\eta_{10}^*\star\eta_{10}\\
   &\phantom{....................................................................................}+\tilde{\eta}_{11}^*\star\eta_{00}+\eta_{11}^*\star\eta_{11}\Big],
\end{align*}
where $\tilde{\eta}_{11}:=\sum\limits_{I,J}\left(-1\right)^{|I|+|J|}a_{IJ}\overline{\Psi}_I\Psi_J\in\cG_M$ if $\eta_{11}:=\sum\limits_{I,J}a_{IJ}\overline{\Psi}_I\Psi_J$ for some $a_{IJ}\in\mathbbm{C}$.
$\tilde{\eta}_{11}$ occurs due to the anticommutations of $\psi_{M+1}$ with $\eta_{11}^*$ and of $\overline{\psi}_{M+1}$ with
$\eta_{11}$ in the second and the fifth term on the r.h.s.\ of (\ref{grogla}), respectively. Observing that 
\begin{align*}
\nonumber &\int\cD_{M}(\overline{\Psi},\Psi)\,\tilde{\eta}_{11}^*\star\tilde{\eta}_{11}\\
\nonumber & \qquad=\sum\limits_{I,J,K,L}a_{IJ}\overline{a}_{LK}\left(-1\right)^{|I|+|J|+|K|+|L|}\int\cD_{M}(\overline{\Psi},\Psi)\left(\overline{\Psi}_I\Psi_J\right)\star\left(\overline{\Psi}_K\Psi_L\right)\\
&\qquad= \int\cD_{M}(\overline{\Psi},\Psi)\,\eta_{11}^*\star\eta_{11},
\end{align*}
since $|I|+|J|+|K|+|L|$ is even (otherwise both integrals vanish), we finally conclude
\begin{align*}
 \nonumber &\left(-1\right)^{|M|+1}\int\cD_{M+1}(\overline{\Psi},\Psi)\left(\xi^*\star\xi\right)\\ 
 \nonumber &\qquad=\left(-1\right)^{|M|+2}\int\cD_{M}(\overline{\Psi},\Psi)\Big[\eta_{00}^*\star\eta_{00}+\left(\eta_{00}+\tilde{\eta}_{11}\right)^*\star\left(\eta_{00}+\tilde{\eta}_{11}\right)
 +\eta_{01}^*\star\eta_{01}+\eta_{10}^*\star\eta_{10}\Big],
\end{align*}
which is non-negative by the induction hypothesis.
\end{proof}
Finally, we can express the trace of an operator of $\mathcal{B}\left(\wedge\cH\right)$ and, thanks to 
Lemma~\ref{prodrule}, the trace of 
a product of such operators as a Grassmann integral.
\begin{thm}
 For all $A\in\mathcal{B}\left(\wedge\cH\right)$ we have
\begin{align}
 \mathrm{tr}_{\wedge\cH}\left\{A\right\}=\left(-1\right)^{|M|}
      \int\cD(\overline{\Psi},\Psi)\,\Theta\left(A\right).\label{TRACE}
\end{align}
\end{thm}
\begin{proof}
 We assume that $A\in\cB\left(\wedge\cH\right)$ is normal-ordered. Due to the linearity of the trace and the Grassmann integral it suffices to
consider $\mathrm{tr}_{\wedge\cH}\left\{\e_{i_1}\cdots\e_{i_m}\v_{j_1}\cdots\v_{j_n}\right\}$, where
$I:=\left\{i_1,\dots,i_m\right\}$ and $J:=\left\{j_1,\dots,j_n\right\}$ are ordered. For $I\neq J$, both the l.h.s. and
the r.h.s. of (\ref{TRACE}) vanish. For $I=J$,
 the l.h.s.\ of (\ref{TRACE}) is given by
\begin{align*}
 \mathrm{tr}_{\wedge\cH}\left\{\e_{i_1}\cdots\e_{i_m}\v_{i_1}\cdots\v_{i_m}\right\}=\left(-1\right)^{\frac{1}{2}|I|\left(|I|-1\right)}
      2^{|M|-|I|}.
\end{align*}
On the r.h.s.\ of (\ref{TRACE}) we have
$\Theta\left(\e_{i_1}\cdots\e_{i_m}\v_{i_1}\cdots\v_{i_m}\right)=\overline{\psi}_{i_1}\cdots\overline{\psi}_{i_m}\psi_{i_1}\cdots \psi_{i_m}$
and, thus,
\begin{align*}
\nonumber \int\cD(\overline{\Psi},\Psi)\,\overline{\psi}_{i_1}\cdots\overline{\psi}_{i_m}\psi_{i_1}\cdots \psi_{i_m}
   & =\left(-1\right)^{\frac{1}{2}|I|\left(|I|+1\right)}\int\cD(\overline{\Psi},\Psi)\,\prod\limits_{\alpha=1}^m\left(\psi_{i_\alpha}\overline{\psi}_{i_\alpha}\right)\\
&=\left(-1\right)^{|M|}\left(-1\right)^{\frac{1}{2}|I|\left(|I|+1\right)}2^{|M|-|I|},
\end{align*}
since $\prod\limits_{\alpha\in I}\left(\psi_{\alpha}\overline{\psi}_{\alpha}\right)\mathrm{e}^{2\left(\overline{\Psi},\Psi\right)}=
\prod\limits_{\alpha\in I}\left(\psi_{\alpha}\overline{\psi}_{\alpha}\right)\prod\limits_{\alpha\in M\backslash I}\left(1+2\overline{\psi}_\alpha\psi_\alpha\right)$ and,
therefore,
\begin{align*}
\prod\limits_{\alpha\in M}\left(\frac{\delta}{\delta\overline{\psi}}\frac{\delta}{\delta\psi}\right)\prod\limits_{\alpha\in I}\left(\psi_{\alpha}\overline{\psi}_{\alpha}\right)\mathrm{e}^{2\left(\overline{\Psi},\Psi\right)}
    =\left(-2\right)^{|M|-|I|}.
\end{align*}
The proof is complete by $\left(-2\right)^{|M|-|I|}=
    \left(-1\right)^{|M|}\left(-1\right)^{|I|}2^{|M|-|I|}$.
\end{proof}
Due to the restriction to a Hilbert space with even dimension, we henceforth skip the factor 
$\left(-1\right)^{|M|}$.


\section{Representability Conditions from Grassmann Integrals}

The last section allows for an application of the Grassmann integration on the problem of representability for fermion
systems. 
In particular, we are interested in necessary conditions for the 1- and 2-pdm to have their origin in a density matrix $\rho$
\cite{BKM}. In the language of Grassmann integration we call the equivalents of density matrices Grassmann densities.
\begin{defn}
 A Grassmann variable $\vartheta^*\star\vartheta\in\cG_M$ is called Grassmann density if it is normalized, i.e., fulfills
\begin{align*}
 \int\cD \left(\overline{\Psi},\Psi\right)\vartheta^*\star\vartheta\,  =1.
\end{align*}
\end{defn}

By definition, the Grassmann density is positive semi-definite and self-adjoint. For a given state $\rho$, the map $\Theta$ immediately provides
$\vartheta^*\star\vartheta$, namely $\vartheta^*\star\vartheta=\Theta\left(\rho\right)$. Thanks to the product rule for $\Theta$ and the positive semi-definiteness
of $\rho$, we also have $\vartheta^*\star\vartheta=\Theta\big(\rho^{\frac{1}{2}}\rho^{\frac{1}{2}}\big)=\Theta\big(\rho^{\frac{1}{2}}\big)
\star\Theta\big(\rho^{\frac{1}{2}}\big)$. $\Theta$ is a bijection and compatible with the involution. This implies that
$\vartheta=\Theta\big(\rho^{\frac{1}{2}}\big)$. 
Given a Grassmann density, we can formulate the problem of representability by Grassmann integrals using the trace-formula (\ref{TRACE}).
\begin{defn}
 Let $\left\{\overline{\psi}_i,\psi_i\right\}_{i \in M}$ be the generators of $\cG_M$ and associate $\left\{\psi_i\right\}_{i\in M}$
with a fixed ONB of $\cH$. The 1-pdm $\gamma_\vartheta\in\cB\left(\cH\right)$ and 2-pdm
$\Gamma_\vartheta\in\cB\left(\cH\otimes\cH\right)$ of a Grassmann density $\vartheta^*\star\vartheta$ are defined by their matrix
elements:
\begin{align}
 \left<\psi_k,\gamma_\vartheta\psi_l\right>&:=\int\cD(\overline{\Psi},\Psi)\,\vartheta^*\star\vartheta\star\overline{\psi}_l
	\star\psi_k\,  \quad\text{and}\label{1RDMGM}\\
 \left<\psi_m\otimes\psi_n ,\Gamma_\vartheta\left(\psi_l\otimes\psi_k\right)\right>&:=\int\cD \left(\overline{\Psi},\Psi\right)\vartheta^*\star\vartheta\star\overline{\psi}_k
	\star\overline{\psi}_l\star\psi_m\star\psi_n\,  \label{2RDMGM}.
\end{align}
\end{defn}

Applying the trace formula (\ref{TRACE}) on (\ref{1RDMGM}) and (\ref{2RDMGM}), respectively, we observe that 
\begin{align*}
  \left<\psi_k,\gamma_\rho\psi_l\right>&=\mathrm{tr}_{\wedge\cH}\left\{\Theta^{-1}\left(\vartheta^*\star\vartheta\right)\e_l\v_k\right\}\quad\text{and}\\
  \left<\psi_m\otimes\psi_n ,\Gamma_\rho\left(\psi_l\otimes\psi_k\right)\right>&=
	  \mathrm{tr}_{\wedge\cH}\left\{\Theta^{-1}\left(\vartheta^*\star\vartheta\right)\e_l\e_k\v_n\v_m\right\},
\end{align*}
which agrees with the common definition of the 1- and 2-pdm \cite{BKM} if we interpret $\Theta^{-1}\left(\vartheta^*\star\vartheta\right)
=\left(\Theta^{-1}\left(\vartheta\right)\right)^*\Theta^{-1}\left(\vartheta\right)$ as a
density matrix $\rho\in\cB\left(\wedge\cH\right)$. The problem of representability can be formulated as follows:
\begin{defn}
 We call $\left(\gamma,\Gamma\right)\in\cB\left(\cH\right)\times\cB\left(\cH\otimes\cH\right)$ representable if there exists a Grassmann density $\vartheta^*\star\vartheta$ such that
$\left(\gamma,\Gamma\right)=\left(\gamma_\vartheta,\Gamma_\vartheta\right)$.
\end{defn}



\subsection{Conditions on the One-Particle Density Matrix}
The lower and upper bound for the eigenvalues of the 1-pdm $\gamma_\vartheta$ of a Grassmann state $\vartheta^*\star\vartheta$ arise directly from the definition of the 1-pdm (see \cite{BKM} for further details).
Here, we would like to derive the conditions by Grassmann integration. To this end, we consider certain subspaces of $\cG_M$.
\begin{defn}	
 For any $n\in\mathbbm{N}$, $n\leq |M|$, we define the subspace
\begin{align*}
 \cG_M^{(n)}:=\mathrm{span}\left\{\overline{\Psi}_I\Psi_J|\ |I|,|J|\leq n\right\}\subseteq\cG_M.
\end{align*}
\end{defn}

Bounds for the 1-pdm rise by considering $\cG_M^{(1)}$. In what follows, we call conditions derived by considering
$\cG_M^{(n)}$ as conditions of $n$-th order.
\begin{lem}
 Theorem~\ref{positivity} implies
\begin{align*}
 \gamma_\vartheta \geq 0.
\end{align*}
\end{lem}
\begin{proof}
Let $\left\{\overline{\psi}_i,\psi_i\right\}_{i\in M}$ be the generators of $\cG_M$ and $\alpha_k\in\mathbbm{C}\ \forall\,k\in M$. In Theorem~\ref{positivity}, we make use of Equation~(\ref{cyclic2}) with 
$\eta:=\phi\star\vartheta^*$ and
    $\phi:=\sum\limits_{k\in M}\alpha_k\psi_k\in\cG_M$. We observe that, according to the involution $\left(\,\cdot\,\right)^*$ on $\cG_M$,
 $\phi^*=\sum\limits_{k\in M}\overline{\alpha}_k\overline{\psi}_k$, and $\eta^*=\left(\phi\star\vartheta^*\right)^*=\vartheta\star\phi^*$. This leads to
\begin{align*}
\nonumber 0&\leq \int\cD(\overline{\Psi},\Psi)\,\eta^*\star\eta\,  \\
\nonumber &=\sum\limits_{k,l\in M}\overline{\alpha}_k\alpha_l\int\cD(\overline{\Psi},\Psi)\,\vartheta^*\star\vartheta\star
	      \overline{\psi}_k\star\psi_l\,  \\
      &=\left<f,\gamma_\vartheta f\right>, 
\end{align*}
where $f:=\sum\limits_{i\in M}\overline{\alpha}_i\psi_i\in\cH$ is arbitrary.
\end{proof}
The upper bound for $\gamma_\vartheta$ is given by another choice of $\eta$.
\begin{lem}
Theorem~\ref{positivity} implies
\begin{align*}
 \gamma_\vartheta \leq \mathbbm{1}.
\end{align*}
\end{lem}
\begin{proof}
The bound can be proven by following the steps of the proof of the lower bound. 
Again, we have $\alpha_k\in\mathbbm{C}\ \forall\,k\in M$ and set  
$\phi^*=\sum\limits_{k\in M}\overline{\alpha}_k\overline{\psi}_k\in\cG_M$ and, this time,
$\eta^*=\left(\phi^*\star\vartheta\right)^*=\vartheta^*\star\phi$.
Before we go on, we observe that by the CAR 
on $\cG_M$ given in (\ref{CARonG}),
\begin{align*}
 \nonumber \phi\star\phi^*&=\sum\limits_{k,l\in M}\alpha_k\overline{\alpha}_l\psi_k\star\overline{\psi}_l
      = \sum\limits_{k\in M}\overline{\alpha}_k\alpha_k
	-\sum\limits_{k,l\in M}\alpha_k\overline{\alpha}_l\overline{\psi}_l\star\psi_k.
\end{align*}
Inserting this into the inequality of Theorem~\ref{positivity} and using the associativity of the
star product, we obtain
\begin{align*}
\nonumber 0&\leq \int\cD(\overline{\Psi},\Psi)\,\eta^*\star\eta\,  \\
\nonumber &=\sum\limits_{k\in M}\left|\alpha_k\right|^2
	      -\sum\limits_{k,l\in M}\overline{\alpha}_l\alpha_k\int\cD(\overline{\Psi},\Psi)\,\vartheta^*\star\vartheta\star\overline{\psi}_l\star\psi_k\,  \\
      &=\left<g,\left(\mathbbm{1}-\gamma_\vartheta\right)g \right>,
\end{align*}
where we have used $\int\cD(\overline{\Psi},\Psi)\,\vartheta^*\star\vartheta\,  =1$ and $g:=\sum\limits_{k\in M}\overline{\alpha}_k\psi_k\in\cH$.
\end{proof}
Considering the subspace $\cG_M^{(1)}$, we can summarize our last two results.
\begin{thm}
 Let $\vartheta\star\vartheta^*$ be a Grassmann density and $\gamma_\vartheta$ its 1-pdm. Then the following statements are equivalent:
\begin{itemize}
 \item[a)] $0\leq \gamma_\vartheta \leq \mathbbm{1}$.
 \item[b)] $\forall \mu\in\cG_M^{(1)}$: $\int\cD(\overline{\Psi},\Psi)\,\vartheta^*\star\vartheta \star\mu\geq 0$.
\end{itemize}
\end{thm}
\begin{proof}
 In Theorem 3.1 of \cite{BKM}, the analogue of this theorem has been shown for polynomials in creation and annihilation operators of degree lower than 
 or equal to two. Thanks to the bijection $\Theta$, we have a one-to-one mapping between the space of polynomials of degree lower than
 or equal to two and $\cG_M^{(1)}$. 
\end{proof}


\subsection{G-, P-, and Q-Condition}

We proceed with representability conditions of second order by considering $\cG_M^{(2)}$ and a star-product of $\overline{\psi}$ and $\psi$, in this 
case, for example $\phi:=\sum\limits_{k,l\in M}\alpha_{kl}\psi_k\star\psi_l\in\cG_M$ with $\alpha_{kl}\in\mathbbm{C}\ \forall\,k,l\in M$. This time, we are interested in conditions on $\Gamma_\vartheta$ and use the Grassmann integration to rewrite the matrix elements of the 2-pdm as in (\ref{2RDMGM}).
The first condition is the P-Condition.
\begin{lem}\label{Pcondi} Theorem~\ref{positivity} implies the P-Condition
\begin{align*}
 \Gamma_\vartheta \geq 0.
\end{align*}
\end{lem}
\begin{proof}
The proof is similar to the one in the last subsection. Setting $\phi:=\sum\limits_{k,l\in M}\alpha_{kl}\psi_k\star\psi_l\in\cG_M$ with
$\alpha_{kl}\in\mathbbm{C}\ \forall\,k,l\in M$, 
$\eta:=\phi\star\vartheta^*$, and  
$\eta^*=\left(\phi\star\vartheta^*\right)^*=\vartheta\star\phi^*$, we arrive at
\begin{align}
\nonumber 0&\leq \int\cD(\overline{\Psi},\Psi)\,\eta^*\star\eta\,  \\
\nonumber &=\sum\limits_{k,l,m,n\in M}\overline{\alpha}_{kl}\alpha_{mn}
	    \int\cD(\overline{\Psi},\Psi)\,\vartheta^*\star\vartheta\star\overline{\psi}_l\star\overline{\psi}_k\star\psi_m\star\psi_n
	    \,  \\
	  &=\left<F,\Gamma_\vartheta F\right>, \label{posi}
\end{align}
where $F:=\sum\limits_{k,l\in M}\overline{\alpha}_{kl}\left(\psi_m\otimes\psi_n\right)\in\cH\otimes\cH$ is arbitrary.
\end{proof}

The Q-Condition is the next representability condition. In order to obtain a convenient formulation of this condition, we use an exchange operator
on $\cB\left(\cH\otimes\cH\right)$ which is defined by $\mathrm{Ex}\left(f\otimes g\right):=g\otimes f$ for $f,g\in\cH$.
\begin{lem} 
 Theorem~\ref{positivity} implies the Q-Condition
\begin{align*}
 \Gamma_\vartheta+\left(\mathbbm{1}-\mathrm{Ex}\right)\left(\mathbbm{1}\otimes\mathbbm{1}-\gamma_\vartheta\otimes
      \mathbbm{1}-\mathbbm{1}\otimes\gamma_\vartheta\right) \geq 0.
\end{align*}
\end{lem}
\begin{proof}
With 
$\phi:=\sum\limits_{k,l\in M}\overline{\alpha}_{kl}\overline{\psi}_k\star\overline{\psi}_l\in\cG_M$, $\alpha_{kl}\in\mathbbm{C}\ \forall\,k,l\in M$, and 
$\eta=\phi\star\vartheta^*$, we have
\begin{align*}
\nonumber 0&\leq \int\cD(\overline{\Psi},\Psi)\,\eta^*\star\eta\,  \\
 &= \sum\limits_{k,l,m,n\in M}\overline{\alpha}_{kl}\alpha_{mn}\int\cD(\overline{\Psi},\Psi)\,\vartheta^*\star\vartheta
		\star\psi_n\star\psi_m\star\overline{\psi}_k\star\overline{\psi}_l.  
\end{align*}
Aiming for an expression in terms of $\Gamma$ and $\gamma$, we establish normal ordering using the CAR:
\begin{align}
\nonumber \psi_n\star\psi_m\star\overline{\psi}_k\star\overline{\psi}_l&=\delta_{mk}\delta_{nl}
      -\delta_{nk}\delta_{ml}+\delta_{nk}\overline{\psi}_{l}\star\psi_m-\delta_{mk}\overline{\psi}_{l}\star\psi_n
      +\delta_{nl}\overline{\psi}_{k}\star\psi_l\\
    &\phantom{=}-\delta_{ml}\overline{\psi}_{k}\star\psi_n-\overline{\psi}_{k}
			  \star\overline{\psi}_{l}\star\psi_n\star\psi_m. \label{aftercar}
\end{align}
As in the proof of Lemma~\ref{Pcondi}, we write an arbitrary  $G\in\cH\otimes\cH$ as
$G:=\sum\limits_{k,l\in M}\alpha_{kl}\left(\psi_k\otimes\psi_l\right)$ for some $\alpha_{kl}\in\mathbbm{C}$. Hence, 
$\sum\limits_{k,l,m,n\in M}\overline{\alpha}_{kl}\alpha_{mn}\delta_{km}\delta_{ln}=\left<G,\mathbbm{1}\,G\right>$ and
$\sum\limits_{k,l,m,n\in M}\overline{\alpha}_{kl}\alpha_{mn}\delta_{kn}\delta_{lm}=\left<G,\mathrm{Ex}\,G\right>$.
With (\ref{1RDMGM}) and (\ref{2RDMGM}), we find
\begin{align*}
\nonumber 0\leq\left<G,\left(\Gamma_\vartheta+\left(\mathbbm{1}-\mathrm{Ex}\right)\left(\mathbbm{1}\otimes\mathbbm{1}-\gamma_\vartheta\otimes
      \mathbbm{1}-\mathbbm{1}\otimes\gamma_\vartheta\right)\right)G\right>
\end{align*}
by evaluating the Grassmann integral $\int\cD(\overline{\Psi},\Psi)\left(\,\cdot\,\right)  $
on the r.h.s.\ of (\ref{aftercar}).
\end{proof}
The last second order representability condition which can be derived by the described method is the (optimal) G-Condition. Deriving this condition by Grassmann integration
requires a choice of $\eta$, that is not as obvious as before. In the following, $\TRh{\,\cdot\,}$ denotes the trace on $\cH$ and
$\TRhh{\,\cdot\,}$ the trace on $\cH\otimes\cH$.
\begin{lem} 
 Theorem~\ref{positivity} implies the G-Condition: 
\begin{align}
   \forall\, A\in\cB\left(\cH\right):\ \TRhh{\left(A^*\otimes A\right)\left(\Gamma_\vartheta+\mathrm{Ex}
    \left(\gamma_\vartheta\otimes \mathbbm{1}\right)\right)}\geq\left|\TRh{A\gamma_\vartheta}\right|^2. \label{GBed}
\end{align}
\end{lem}
\begin{proof}
 This time, we choose $\eta:=\bigg(\sum\limits_{k,l\in M}\alpha_{kl}\overline{\psi}_k\star\psi_l-c\bigg)\star\vartheta$
with $c:=\sum\limits_{k,l\in M}\alpha_{kl}\int\cD(\overline{\Psi},\Psi)\,\vartheta^*\star\vartheta \star
	\overline{\psi}_k\star\psi_l$ and $\alpha_{kl}\in\mathbbm{C}\ \forall\,k,l\in M$. Before
we apply Theorem~\ref{positivity}, we emphasize that by the CAR
\begin{align}
\nonumber&\bigg(\sum\limits_{k,l\in M}\alpha_{kl}\overline{\psi}_k\psi_l-c\bigg)^*
	\star\bigg(\sum\limits_{k,l\in M}\alpha_{kl}\overline{\psi}_k\psi_l-c\bigg)\\
\nonumber &\quad\qquad\qquad =\overline{c}c-c\sum\limits_{k,l\in M}\overline{\alpha}_{kl}\overline{\psi}_l\star\psi_k-\overline{c}
	\sum\limits_{m,n\in M}\alpha_{mn}\overline{\psi}_m\star\psi_n\\
  &\qquad\qquad\qquad\quad-\sum\limits_{k,l\in M}\overline{\alpha}_{kl}\alpha_{mn}\overline{\psi}_l\star\overline{\psi}_m
	\star\psi_k\star\psi_n+\sum\limits_{k,l,n\in M}\overline{\alpha}_{kl}\alpha_{kn}\overline{\psi}_l\star\psi_n. \label{ASLG}
\end{align}
We consider the last two lines separately and integrate. The integration of the line before the last line in (\ref{ASLG}) yields
\begin{align}
\nonumber &\int\cD(\overline{\Psi},\Psi)\,\vartheta^*\star\vartheta\star\bigg(\overline{c}c-c\sum\limits_{k,l\in M}
	\overline{\alpha}_{kl}\overline{\psi}_l\star\psi_k-\overline{c}\sum\limits_{m,n\in M}\alpha_{mn}\overline{\psi}_m\star\psi_n\bigg)
	  \\
&\qquad\qquad\quad =\overline{c}c-c\overline{c}-\overline{c}c=-\overline{c}c, \label{P1a}
\end{align}
which follows from the definition of $c$. It is important to notice that $c$ does not depend on $\psi$ or $\overline{\psi}$ and, therefore, is
a constant with respect to the Grassmann integration. In detail, we have for $c$:
\begin{align}
c=\sum\limits_{k,l\in M}\alpha_{kl}\int\cD(\overline{\Psi},\Psi)\,
	  \vartheta^*\star\vartheta\star\overline{\psi}_k\star\psi_l=\TRh{A\gamma_\vartheta}, \label{P1b}
\end{align}
if we set $\left<\psi_k,A\psi_l\right>:=\alpha_{kl}$ for every $k,l\in M$ and $A\in\cB\left(\cH\right)$. The evaluation of the Grassmann integral of the last line in (\ref{ASLG})
provides
\begin{align}
\nonumber& -\sum\limits_{k,l\in M}\overline{\alpha}_{kl}\alpha_{mn}\int\cD(\overline{\Psi},\Psi)\,
      \vartheta^*\star\vartheta\star\overline{\psi}_l\star\overline{\psi}_m\star\psi_k\star\psi_n\,
	  \\
\nonumber &\qquad +\sum\limits_{k,l,n\in M}\overline{\alpha}_{kl}\alpha_{kn}\int\cD(\overline{\Psi},\Psi)\,
      \vartheta^*\star\vartheta\star\overline{\psi}_l\star\psi_n\,  \\
&\qquad\qquad\qquad=\TRhh{\left(A^*\otimes A\right)\left(\Gamma_\vartheta+\mathrm{Ex}
    \left(\gamma_\vartheta\otimes \mathbbm{1}\right)\right)}. \label{P2}
\end{align}
Summing up, calculation (\ref{P1a}) together with (\ref{P1b}) and (\ref{P2}) gives
\begin{align*}
 \TRhh{\left(A^*\otimes A\right)\left(\Gamma_\vartheta+\mathrm{Ex}
    \left(\gamma_\vartheta\otimes \mathbbm{1}\right)\right)}-\left|\TRh{A\gamma_\vartheta}\right|^2\geq 0,
\end{align*}
according to Theorem~\ref{positivity}.
\end{proof}
We summarize our results using $\cG_M^{(2)}$:
\begin{thm}
  Let $\vartheta\star\vartheta^*$ be a Grassmann density, $\gamma_\vartheta$ its 1-pdm, and $\Gamma_\vartheta$ its 2-pdm. Then the following statements are equivalent:
\begin{itemize}
 \item[a)] $\left(\gamma_\vartheta,\Gamma_\vartheta\right)$ fulfills $0\leq \gamma_\vartheta\leq \mathbbm{1}$ and the G-, P-, and Q-Conditions.
 \item[b)] $\forall \mu\in\cG_M^{(2)}$: $\int\cD(\overline{\Psi},\Psi)\,\vartheta^*\star\vartheta \star\mu\geq 0$.
\end{itemize}
\end{thm}
\begin{proof}
 Again, we use Theorem 3.1 of \cite{BKM} and the bijection property of $\Theta$, which ensures that the space of 
 polynomials of degree lower or equal than four in creation and annihilation operators is mapped one-to-on to $\cG_M^{(2)}$.
\end{proof}


\subsection{$\mathrm{T}_1$- and Generalized $\mathrm{T}_2$-Condition}

The last sections imply that further conditions on $\gamma_\vartheta$ and $\Gamma_\vartheta$ can be found by taking into account
monomials of higher order of the form $\overline{\psi}_{i_1}\star\cdots\star\overline{\psi}_{i_n}\star\psi_{j_1}\star\cdots\star\psi_{j_n}$ for $n>2$.
Here we face the problem that monomials with $n>2$ have to ``decompose'' into monomials with $n\leq 2$. Due to this, only
some choices of higher order monomials are suitable to derive further representability conditions. One of such monomials is 
given by 
\begin{align*}
 \tau_1:=\sum\limits_{i,j,k\in M}T_{ijk}\psi_i\star\psi_j\star\psi_k\in\cG_M,
\end{align*}
Where $T_{ijk}\in\mathbbm{C}$ is, due to $\left\{\psi_i,\psi_j\right\}_{\star}=0$, totally antisymmetric, i.e., 
$T_{ijk}=-T_{jik}=T_{jki}$. The $\mathrm{T}_1$-Condition is the following.
\begin{thm}
Let $T_q\in\mathcal{B}\left(\cH\right)$ be trace class, and set $F_{T_{q}}:=\sum\limits_{k,n\in M} 
  \overline{T}_{kq n}\left(\varphi_k\otimes\varphi_n\right)\in\cH\otimes\cH$, $T_{kq n}:=\left[T_q\right]_{kn}$.
 Then Theorem~\ref{positivity} implies the $\mathrm{T}_1$-Condition:
\begin{align*}
\sum\limits_{q\in M}\left(2\TRh{\left|T_q\right|^2}-6\TRh{\left|T_q\right|^2\gamma_\vartheta}
   +3\left<F_{T_q},\Gamma_\vartheta F_{T_q}\right>\right)\geq 0.
\end{align*}
\end{thm}
\begin{proof}
 We begin by considering the anticommutator 
$\left\{\tau_1^*,\tau_1\right\}_{\star}\in\cG_M$ and observe that, by construction,
$\left\{\tau_1^*,\tau_1\right\}_{\star}\geq 0$. Furthermore, we can use the CAR to establish normal order in 
$\left\{\tau_1^*,\tau_1\right\}_{\star}$. The ${i,j}-$th matrix element of $A\in\cB\left(\cH\right)$ is denoted by $\left[A\right]_{ij}:=\left<\psi_i , A\psi_j\right>$. Using the antisymmetry of $T_{ijk}$ we arrive at
\begin{align*}
\nonumber \left\{\tau_1^*,\tau_1\right\}_{\star}&=9\sum\limits_{l\in M}\sum\limits_{i,j,m,n\in M} \overline{T}_{ljm}T_{lin}
      \overline{\psi}_m\star\overline{\psi}_j\star\psi_i\star\psi_n\\
\nonumber   &\ \ \ +18\sum\limits_{m,l\in M}\sum\limits_{k,n\in M} 
	    \overline{T}_{kml}T_{lmn}\overline{\psi}_k\star\psi_n+6\sum\limits_{l,m,n\in M}\overline{T}_{lmn}M_{lmn}\\
\nonumber  &=9\sum\limits_{q\in M}\sum\limits_{i,j,m,n\in M} \left[T_q^*\right]_{mj}\left[T_q\right]_{in}
      \overline{\psi}_m\star\overline{\psi}_j\star\psi_i\star\psi_n\\
    &\ \ \ -18\sum\limits_{q\in M}\sum\limits_{k,n\in M} 
	   \left[T_q^*T_q\right]_{kn}\overline{\psi}_k\star\psi_n+6\sum\limits_{q\in M}\TRh{\left|T_q\right|^2}.
\end{align*}
Since $\left\{\tau_1^*,\tau_1\right\}_{\star}\geq 0$, we have by Theorem~\ref{positivity}
\begin{align*}
\int\cD(\overline{\Psi},\Psi)\,\vartheta\star\left\{\tau_1^*,\tau_1\right\}_{\star}\star\vartheta^*\,
        \geq 0.
\end{align*}
Together with (\ref{2RDMGM}), the latter calculations and this positivity of the integral bring us to
\begin{align*}
\nonumber 0 \leq & \,3\sum\limits_{q\in M}\sum\limits_{i,j,m,n\in M} \left[T_q^*\right]_{mj}\left[T_q\right]_{in}
      \left<\psi_i\otimes\psi_n,\Gamma_\vartheta\left(\psi_j\otimes\psi_m\right)\right>\\
    &-6\sum\limits_{q\in M}\sum\limits_{k,n\in M} 
	   \left[\left|T_q\right|^2\right]_{kn}\left<\psi_n,\gamma_\vartheta\psi_k\right>+2\sum\limits_{q\in M}
	  \TRh{\left|T_q\right|^2}.
\end{align*}
Together with $\left<\psi_i,T_q\psi_j\right>=:\left[T_q\right]_{ij}$ and $F_{T_q}:=\sum\limits_{k,n\in M} 
  \overline{T}_{kq n}\left(\varphi_k\otimes\varphi_n\right)$, this yields the assertion.
\end{proof}
The generalized $\mathrm{T}_2$-Condition can be derived equivalently by another choice of $\tau$. Using the anticommutator with a 
combination of two $\overline{\psi}$'s and one $\psi$ (or vice versa), we have three different possibilities: $\tau_{2a}:=
\sum\limits_{i,j,k\in M}T_{ijk}^{(a)}\overline{\psi}_i\star\overline{\psi}_j\star\psi_k$, $\tau_{2b}:=
\sum\limits_{i,j,k\in M}T_{ijk}^{(b)}\overline{\psi}_i\star\psi_j\star\overline{\psi}_k$, and 
$\tau_{2c}:=\sum\limits_{i,j,k\in M}T_{ijk}^{(c)}\psi_i\star\overline{\psi}_j\star\overline{\psi}_k$. A generalization
of these possibilities is given by 
\begin{align*}
  \tau_2:=\sum\limits_{i,j,k\in M}T_{ijk}\overline{\psi}_i\star\overline{\psi}_j
      \star\psi_k+\sum\limits_{i\in M}a_i\overline{\psi}_i,
\end{align*}
where $T_{ijk},\ a_i\in\mathbbm{C}\ \forall\ i,j,k\in M$. This is a generalization, since we obtain $\tau_2=\tau_{2a}$ for 
$\alpha_i\equiv 0$ and $T_{ijk}\equiv T_{ijk}^{(a)}$, $\tau_2=\tau_{2b}$ for $a_i=\sum\limits_{j\in M}T_{ijj}^{(b)}$ and 
$T_{ijk}=-T_{ikj}^{(b)}$, and, finally,
$\tau_2=\tau_{2c}$ for $a_{i}=\sum\limits_{j\in M}\left(T_{jji}^{(c)}-T_{jij}^{(c)}\right)$ and
$T_{ijk}=T_{kij}^{(c)}$. The identities can be seen by using the CAR. Unfortunately, if one uses the generalization $\tau_2$, 
symmetry properties on $T_{ijk}$ like, for example, $T_{ijk}^{(a)}=-T_{jik}^{(a)}$ in $\tau_{2a}$ or $T_{ijk}^{(c)}=-T_{ikj}^{(c)}$ in
$\tau_{2c}$ vanish.
The generalized $\mathrm{T}_2$-Condition rises from $\left\{\tau_2^*,\tau_2\right\}_\star\geq 0$. In order to state the condition in a compact form,
we need some new notation.
\begin{defn}\label{DefT2}
 For $T_k\in\cB\left(\cH\right)$, $\left[T_k\right]_{ij}:=T_{ijk}\ \forall i,j,k\in M$, and $\underline{a}\in\mathbbm{C}^{|M|}$, 
  we define $G_{M_k}\in\cH\otimes\cH$ and
 the matrices $Q_1\in\cB\left(\cH\otimes\cH\right)$ and $Q_2,\ Q_3\in\cB\left(\cH\right)$ by
\begin{align*}
&G_{M_k}:=\sum\limits_{i,j\in M}\left[T_k\right]_{ij}\left(\psi_i\otimes\psi_j\right),\\
& \left<\psi_k\otimes\psi_m,Q_1\left(\psi_n\otimes\psi_j\right)\right>:=\left[\overline{T}_k^{(A)}T_n^{(A)}\right]_{jm},\\
&\left<\psi_i,Q_2\psi_j\right>:=\TRh{\left(T_i^{(A)}\right)^*T_j},\\
&\left<\psi_i,Q_3\psi_j\right>:=\sum\limits_{q\in M}\left(\left[\left(T_i^{(A)}\right)^*\right]_{jq}a_q+\left[T_j^{(A)}\right]_{iq}\overline{a}_q\right),
\end{align*}
where $\left[T_k^{(A)}\right]_{ij}:=\frac{1}{2}\left(\left[T_k\right]_{ij}-\left[T_k\right]_{ji}\right)=-\left[T_k^{(A)}\right]_{ji}$ is
the antisymmetric part of $T_k$.
\end{defn}
\begin{thm}
 Let $T_k$, $\underline{a}$, $G_{T_q}$ and $Q_1,\ Q_2,\ Q_3$ be as in Definition \ref{DefT2}.
 Then Theorem~\ref{positivity} implies the generalized $\mathrm{T}_2$-Condition:
\begin{align*}
 \sum\limits_{q\in M}\left<G_{T_q},\Gamma_\vartheta G_{T_q}\right>+4\TRhh{Q_1\Gamma_\vartheta}+2\TRh{\left(Q_2+Q_3\right)\gamma_\vartheta}+\left|\underline{a}\right|^2\geq 0.
\end{align*}
\end{thm}
\begin{proof}
The first task is to bring $\left\{\tau_2^*,\tau_2\right\}$ into normal order. Afterwards, the two terms of third order cancel. Only
terms of order less than or equal to two remain. To calculate the anticommutator we use 
$\left\{\left(\mu+\eta\right)^*,\mu+\eta\right\}_\star=\left\{\mu^*,\mu\right\}_\star+2\mathfrak{Re}\left\{\mu^*,\eta\right\}_\star
+\left\{\eta^*,\eta\right\}_\star$ for $\mu:=\sum\limits_{i,j,k\in M}T_{ijk}\overline{\psi}_i\star\overline{\psi}_j
      \star\psi_k$ and $\eta:=\sum\limits_{i\in M}a_i\overline{\psi}_i$. By the CAR, we have
\begin{align*}
 \left\{\eta^*,\eta\right\}_\star=\sum\limits_{i\in M}\left|a_i\right|^2, \qquad 
    \left\{\mu^*,\eta\right\}_\star=\sum\limits_{k,n\in M}\sum\limits_{q\in M}\left(\overline{T}_{qnk}
      -\overline{T}_{nqk}\right)a_q\overline{\psi}_k\star\psi_n,
\end{align*}
and
\begin{align*}
\nonumber \left\{\mu^*,\mu\right\}_\star=&\sum\limits_{j,k,m,n\in M}\sum\limits_{q\in M}\Big(\left(\overline{T}_{jqk}-\overline{T}_{qjk}\right)
		    \left(T_{qmn}-T_{mqn}\right)+\overline{T}_{njq}T_{kmq}\Big)\overline{\psi}_k\star\overline{\psi}_m\star\psi_j\star\psi_n\\
  &+\sum\limits_{k,n\in M}\sum\limits_{p,q\in M}\left(\overline{T}_{pqk}-\overline{T}_{qpk}\right)T_{pqn}\overline{\psi}_k\star\psi_n.
\end{align*}
We set $T_{ijq}=:\left[T_q\right]_{ij}$ where $T_q\in\cB\left(\cH\right)\ \forall q\in M$ and observe that $\left[\overline{T}_q\right]_{ij}=\left[T_q^*\right]_{ji}$,
$\overline{T}_{qnk}-\overline{T}_{nqk}=2\left[\left(T_k^{(A)}\right)^*\right]_{nq}$, and $T_{qmn}-T_{mqn}=2\left[T_n^{(A)}\right]_{qm}$, 
where $T^{(A)}$ is the antisymmetric part of $T$ (see Definition \ref{DefT2}). This allows us to rewrite the anticommutators:
\begin{align*}
 2\mathfrak{Re}\left\{\mu^*,\eta\right\}_\star=2\sum\limits_{k,n\in M}\sum\limits_{q\in M}\left(\left[\left(T_k^{(A)}\right)^*\right]_{nq}a_q
	+\left[T_n^{(A)}\right]_{qk}\overline{a}_q\right)\overline{\psi}_k\star\psi_n
\end{align*}
and
\begin{align}
 \nonumber \left\{\mu^*,\mu\right\}_\star=&\sum\limits_{j,k,m,n\in M}\sum\limits_{q\in M}\bigg(4\left[\left(T_k^{(A)}\right)^*\right]_{qj}
		    \left[T_n^{(A)}\right]_{qm}+\left[T^*_q\right]_{jn}\left[T_q\right]_{km}\bigg)\overline{\psi}_k\star\overline{\psi}_m\star\psi_j\star\psi_n\\
  &+2\sum\limits_{k,n\in M}\sum\limits_{p,q\in M}\left[\left(T_k^{(A)}\right)^*\right]_{qp}\left[T_n\right]_{pq}\overline{\psi}_k\star\psi_n.\label{Muemue}
\end{align}
In the next step we use $\left<\psi_i, A \psi_j\right>=\left[A\right]_{ij}$ for $A\in\cB\left(\cH\right)$ and the Grassmann representation of $\gamma$ and $\Gamma$ from
(\ref{1RDMGM}) and (\ref{2RDMGM}). Definition \ref{DefT2} then leads to
\begin{align*}
\nonumber &\sum\limits_{j,k,m,n\in M}\sum\limits_{q\in M}\left[T^*_q\right]_{jn}\left[T_q\right]_{km}\int\cD(\overline{\Psi},\Psi)\,
      \vartheta^*\star\vartheta\star\overline{\psi}_k\star\overline{\psi}_m\star\psi_j\star\psi_n = \sum\limits_{q\in M}\left<G_{T_q},\Gamma_\vartheta G_{T_q}\right>
\end{align*}
for $G_{T_q}:=\sum\limits_{i,j\in M}\left[T_q\right]_{ij}\left(\psi_i\otimes\psi_j\right)\in\cH\otimes\cH$. Moreover, we have with
$\left<\psi_m\otimes\psi_k,Q_1\left(\psi_j\otimes\psi_n\right)\right>:=\left[\overline{T}_k^{(A)}T_n^{(A)}\right]_{jm}$
\begin{align*}
\nonumber &4\sum\limits_{j,k,m,n,q\in M}\left[\overline{T}_k^{(A)}\right]_{jq}
		    \left[T_n^{(A)}\right]_{qm}\int\cD(\overline{\Psi},\Psi)\,
      \vartheta^*\star\vartheta\star\overline{\psi}_k\star\overline{\psi}_m\star\psi_j\star\psi_n = 4\TRhh{Q_1\Gamma_\vartheta}.
\end{align*}
Furthermore,
\begin{align*}
2\sum\limits_{k,n\in M}\sum\limits_{p,q\in M}\left[\left(T_k^{(A)}\right)^*\right]_{qp}\left[T_n\right]_{pq}\int\cD(\overline{\Psi},\Psi)\,
      \vartheta^*\star\vartheta\star\overline{\psi}_k\star\psi_n=2\TRh{Q_2\gamma_\vartheta}
\end{align*}
for $\left[Q_2\right]_{kn}:=\TRh{\left(T_k^{(A)}\right)^*T_n}$. Finally, with 
$\left[Q_3\right]_{ij}:=\sum\limits_{q\in M}\left(\left[\left(T_i^{(A)}\right)^*\right]_{jq}a_q+\left[T_j^{(A)}\right]_{qi}\overline{a}_q\right)$
we have 
\begin{align*}
 2\mathfrak{Re}\int\cD(\overline{\Psi},\Psi)\,
      \vartheta^*\star\vartheta\star\left\{\mu^*,\eta\right\}_\star  =2\TRh{Q_3\gamma_\vartheta}.
\end{align*}
$\sum\limits_{i}\left|a_i\right|^2=:\left|\underline{a}\right|^2$ is the squared unitary norm
of $\underline{a} $. The proof is complete by inserting the latter calculations into the inequality of Theorem~\ref{positivity}.
\end{proof}
As already mentioned, we have antisymmetry properties for certain choices of $\underline{a}$ and $T_{ijk}$. In $\tau_{2a}$, which we
gain by setting $\underline{a}\equiv 0$ and $T_{ijk}=T_{ijk}^{(a)}=\left[T_k^{(a)}\right]_{ij}$, we have $\left[T_k\right]_{ij}
=-\left[T_k\right]_{ji}$ or $T_k\equiv T_k^{(A)}$. In this case, we have a simplification of the generalized $\mathrm{T}_2$-Condition:
\begin{cor}
 For $\underline{a}\equiv 0$, $T_k\equiv T_k^{(A)}$, $\big[\tilde{T}_k\big]_{ij}:=\left[T_j\right]_{ik}$, we have the $\mathrm{T}_{2a}$-Condition given by
\begin{align*}
 \sum\limits_{q\in M} \left(\big< G_{\tilde{T}_{q}},\Gamma_\vartheta G_{\tilde{T}_{q}}\big>
  +4\TRhh{\left(\tilde{T}_q^*\otimes \tilde{T}_q\right)\Gamma_\vartheta}+2\TRh{\big|\tilde{T}_q\big|^2\gamma_\vartheta}\right)\geq 0.
\end{align*}
\end{cor}
 \begin{proof}
 With $\underline{a}\equiv 0$ we only have to consider $\left\{\mu^*,\mu\right\}_\star$ and can use (\ref{Muemue}) with 
 $T_k\equiv T_k^{(A)}$.
\end{proof}
 We can also use an antisymmetry property in  $\tau_{2c}$ which leads to a condition $\mathrm{T}_{2c}$. Unfortunately, there is no
simplification compared to the generalized $\mathrm{T}_2$-Condition. There is, however, no antisymmetry property
in $\tau_{2b}$.

Since $\left\{\tau^*_{1},\tau_{1}\right\}_\star,\left\{\tau^*_{2},\tau_{2}\right\}_\star\in\cG_M^{(3)}$, the $\mathrm{T}_1$- and $\mathrm{T}_2$-Conditions are conditions
of third order.


\section{Quasifree Grassmann States}

The notion of Grassmann integration allows for a calculation of traces on the fermion Fock space by Grassmann integrals 
and, in turn, to reformulate representability condition in terms of Grassmann integrals. At last, we consider
quasifree states, their one-particle density matrices, and the expression of their relation in terms of Grassmann integrals.

In the following, we will abbreviate the expectation value of a Grassmann variable $\mu\in\cG_M$ with respect to a Grassmann density $\varkappa\in\cG_M$
by 
\begin{align*}
 \int\cD(\overline{\Psi},\Psi)\,\varkappa\star\mu =:\left<\,\mu\,\right>_\varkappa.
\end{align*}

\begin{defn}\label{defquasta}
 Let $N\in\mathbbm{N}$ and $\widetilde{\psi}_i$ denote either $\psi_i\in\cG_M$ or $\overline{\psi}_i\in\cG_M$, where $\left\{\overline{\psi}_i,\psi_i\right\}_{i\in M}$
is a set of generators of $\cG_M$. We call a Grassmann density $\varkappa$ quasifree if
\begin{itemize}
 \item[1)]  $\left<\widetilde{\psi}_1\star\widetilde{\psi}_2\star\dots\star\widetilde{\psi}_{2N-1}\right>_\varkappa=0$ and
 \item[2)] $\left<\widetilde{\psi}_1\star\widetilde{\psi}_2\star\dots\star\widetilde{\psi}_{2N}\right>_\varkappa
	     ={\sum\limits_{\pi}}'\left(-1\right)^{\pi}\left< \widetilde{\psi}_{\pi\left(1\right)}\star\widetilde{\psi}_{\pi\left(2\right)}\right>_\varkappa\times\cdots\times\left<\widetilde{\psi}_{\pi\left(2N-1\right)}\star\widetilde{\psi}_{\pi\left(2N\right)}\right>_\varkappa$,
\end{itemize}
where ${\sum\limits_{\pi}}'$ denotes the sum over all permutations $\pi$ obeying $\pi(1)<\pi(3)<\dots<\pi(2N-1)$
and $\pi(2j-1)<\pi(2j)$ for all $1\leq j\leq N$.
The maximal number of (distinct) $\psi_i$ or $\overline{\psi}_i$ in 1) and 2) is less or equal $|M|$. 
\end{defn}
\begin{rem}
 We have to restrict $N$ in the latter definition or extend $M$ sufficiently, since the expression on the l.h.s.\ 
 of condition 1) and 2), respectively, vanishes, if the number of $\psi_i$ or $\overline{\psi}_i$ is larger 
than $|M|$.
\end{rem}

As it is already known from \cite{BLS}, there is a unique characterization of quasifree states by the 1-pdm. In detail, 
assuming particle number-conservation and defining
\begin{align*}
 \widetilde{\gamma}:=\begin{pmatrix} \gamma & 0 \\ 0 & \mathbbm{1}-\overline{\gamma} \end{pmatrix}\in\cB\left(\cH\oplus\cH\right),
\end{align*}
which is the generalized 1-pdm corresponding to $\gamma$, one has the following theo\-rem. 
\begin{thm}\label{sedra}
 Let $\widetilde{\gamma}=\begin{pmatrix} \gamma & 0 \\ 0 & \mathbbm{1}-\overline{\gamma} \end{pmatrix}$ be an operator on $\cH\oplus\cH$ with $\TRh{\gamma}<\infty$ and $0\leq \widetilde{\gamma}\leq \mathbbm{1}$. Then there
exists a unique quasifree state $\rho$ with $\mathrm{tr}_{\wedge\cH}\left\{\rho,\widehat{\mathbbm{N}}\right\}<\infty$ such that $\widetilde{\gamma}=\widetilde{\gamma}_\rho$.
\end{thm}

 For a proof see \cite{BLS}.\ \\

In the language of Grassmann integration, the reverse direction, namely that $\widetilde{\gamma}_\varkappa$, i.e., the generalized 1-pdm of a quasifree Grassmann density $\varkappa$, has to fulfill $0\leq \widetilde{\gamma}_\varkappa\leq \mathbbm{1}$, can 
be deduced by appropriate choices of $\phi\in\cG_M$ in the positivity condition 
\begin{align*}
 \left<\phi^*\star\phi\right>_\varkappa\geq 0.
\end{align*}
The aim of this section is to determine the unique quasifree Grassmann density subject to Theorem~\ref{sedra}, i.e., the element of a Grassmann algebra corresponding the state given in \cite{BLS}. To this end, we 
consider an operator $\widetilde{\gamma}\in\cB\left(\cH\oplus\cH\right)$ with $0\leq\widetilde{\gamma}\leq \mathbbm{1}$ and its eigenvalues
$\lambda_i$ and $(1-\lambda_i)$, where $0\leq \lambda_i\leq\frac{1}{2}$, $i\in M$. Furthermore, we define $P_0$ to be the projection onto the subspace of $\wedge \cH$ on which
$\sum\limits_{i:\lambda_i=0}\e_i\v_i=0$ for $i\in M$. Moreover, for any $i\in M$ the quantity $q_i$ is given by the relation $\left(1+\mathrm{e}^{q_i}\right)^{-1}=\lambda_i$.
Then, according to \cite{BLS}, any operator $\widetilde{\gamma}$ with $0\leq\widetilde{\gamma}\leq\mathbbm{1}$ is the generalized 1-pdm
of a unique quasifree state $\rho\in\cB\left(\wedge\cH\right)$ given by
\begin{align}
 \rho:=\frac{G}{\mathrm{tr}_{\wedge\cH}\left\{G\right\}}, \label{qfs}
\end{align}
where 
\begin{align*}
 G:=P_0\mathrm{e}^{-H}\quad\text{and}\quad
 H:=\sum\limits_{i:\lambda_i\neq 0}q_i \e_i\v_i.
\end{align*}
Before we turn to the definition of the Grassmann density corresponding to (\ref{qfs}), we introduce the abbreviations
$\Theta_0:=\Theta\left(P_0\right)\in\cG_M$ and ${\prod\limits_{i=1}^n}^\star \mu_i:=\mu_1\star\mu_2\star\cdots\star\mu_n$ for $\mu_i\in\cG_M$.
Furthermore, we associate the generators $\left\{\overline{\psi}_i,\psi_i\right\}_{i\in M}$ of $\cG_M$ with the ONB $\left\{\psi_i\right\}_{i\in M}$ of $\cH$,
where the $\psi_i$ are the eigenvectors of $\gamma$ corresponding to the eigenvalues $\lambda_i$ and $\left(1-\lambda_i\right)$.
\begin{lem}\label{defstaqf}
 Let $\left\{\psi_i\right\}_{i\in M}$ be an ONB of $\cH$ such that $\gamma\psi_i=\lambda_i\psi_i$ and let $\cG_M$ be generated by $\left\{\overline{\psi}_i,\psi_i\right\}_{i\in M}$.  The Grassmann density $\varkappa\in\cG_M$ corresponding to $\rho=\frac{G}{\mathrm{tr}_{\wedge\cH}\left\{G\right\}}$ is given by
 \begin{align}
  \varkappa=\frac{1}{Z}
	      \left(\Theta_0\star{\prod\limits_{i:\lambda_i\neq0}}^\star\left(\left(\mathrm{e}^{-q_i}-1\right)\overline{\psi}_i\psi_i+1\right)\right),\label{gmqfstate}
 \end{align}
where 
\begin{align*}
 Z:=\int\cD(\overline{\Psi},\Psi)\,\Theta_0\star{\prod\limits_{i:\lambda_i\neq0}}^\star\left(\left(\mathrm{e}^{-q_i}-1\right)\overline{\psi}_i\psi_i+1\right). 
\end{align*}
\end{lem}
\begin{proof}
 We consider $\Theta\left(\rho\right)$ subject to (\ref{qfs}). First, we observe that 
$\e_i\v_i$ commutes with $\e_k\v_k$ for every $i,k$. Therefore, we have
\begin{align*}
 \mathrm{e}^{-H}=\prod\limits_{i:\lambda_i\neq 0}\left(\sum\limits_{n=1}^\infty\frac{\left(-q_i\right)^n}{n!}\e_i\v_i+1\right)
	    =\prod\limits_{i:\lambda_i\neq 0}\left(\left(\mathrm{e}^{-q_i}-1\right)\e_i\v_i+1\right),
\end{align*}
since $\left(\e_i\v_i\right)^n=\e_i\v_i$. Thus,
\begin{align*}
 \Theta\left(P_0\mathrm{e}^{-H}\right)&=\Theta_0\star\Theta\left(\prod\limits_{i:\lambda_i\neq 0}\left(\left(\mathrm{e}^{-q_i}-1\right)\e_i\v_i+1\right)\right)
	=\Theta_0\star{\prod\limits_{i:\lambda_i\neq0}}^\star\left(\left(\mathrm{e}^{-q_i}-1\right)\overline{\psi}_i\psi_i+1\right),
\end{align*}
where we have used that $\Theta\left(AB\right)=\Theta\left(A\right)\star\Theta\left(B\right)$.
\end{proof}
The Grassmann state corresponding to the Grassmann density (\ref{gmqfstate}) is given by the map 
\begin{align*}
 \cG_M\to\mathbbm{C},\quad\mu\mapsto \left<\,\mu\,\right>_\varkappa.
\end{align*}

We want to check that the Grassmann density from Lemma~\ref{defstaqf} is quasifree, i.e., fulfills conditions 1) and 2) 
from Definition~\ref{defquasta}. The uniqueness of $\varkappa$ follows from the bijection property of the map $\Theta$.
\begin{thm}
 The Grassmann state $\varkappa$ in Lemma~\ref{defstaqf} is quasifree. 
\end{thm}
\begin{proof}
 We consider the state 
 \begin{align*}
 \varkappa_\mu:={\prod\limits_{i\in M}}^\star\left(r_i\overline{\psi}_i\psi_i+1\right),
\end{align*}
where $r_i:= \mathrm{e}^{-q_i\left(\mu\right)}-1$ and $q_i\left(\mu\right)\equiv \mu\in\mathbbm{R}$ for all $i$ with $\lambda_i=0$ and
$q_i\left(\mu\right)\equiv q_i$ for all $i$ with $\lambda_i\neq 0$. The quasifreeness of $\varkappa$ follows by the quasifreeness of $\varkappa_\mu$
and a limiting argument. The first claim of Definition~\ref{defquasta} is immediate for $\varkappa_\mu$, since the Grassmann integral
vanishes for an odd number of $\widetilde{\psi}$'s. This can be seen by Remark~\ref{intlsg} and the chequerboad. The validity of Equation 2) of Definition~\ref{defquasta} has already been proven in \cite{GAU}. Here we emphasize the main steps and transfer the notation of \cite{GAU}
to Grassmann Integrals. We consider the l.h.s.\
of claim 2) of Definition~\ref{defquasta},
\begin{align*}
 \left<\widetilde{\psi}_a\star\widetilde{\psi}_b\star\widetilde{\psi}_c\star\cdots\star\widetilde{\psi}_f\right>_{\varkappa_\mu}
=\int\cD(\overline{\Psi},\Psi)\,\varkappa_\mu\star\widetilde{\psi}_a\star\widetilde{\psi}_b\star\widetilde{\psi}_c\star\cdots\star\widetilde{\psi}_f,
\end{align*}
with $2N$ generators $\widetilde{\psi}_a,\cdots,\widetilde{\psi}_f$. In the first step we eliminate $\widetilde{\psi}_a$ from the
expectation value by a pull through formula. To this end we use $\left\{\widetilde{\psi}_a,\widetilde{\psi}_b\right\}_\star:=
\widetilde{\psi}_a\star\widetilde{\psi}_b+\widetilde{\psi}_b\star\widetilde{\psi}_a$, which is either $1$, $-1$ or $0$. This yields
\begin{align*}
 &\left<\widetilde{\psi}_a\star\widetilde{\psi}_b\star\widetilde{\psi}_c\star\cdots\star\widetilde{\psi}_f\right>_{\varkappa_\mu}\\
    &\qquad=\left\{\widetilde{\psi}_a,\widetilde{\psi}_b\right\}_\star\left<\widetilde{\psi}_c\star\widetilde{\psi}_d\star\cdots\star\widetilde{\psi}_f\right>_{\varkappa_\mu}-\left\{\widetilde{\psi}_a,\widetilde{\psi}_c\right\}_\star\left<\widetilde{\psi}_b\star\widetilde{\psi}_d\star\cdots\star\widetilde{\psi}_f\right>_{\varkappa_\mu}\\
    &\qquad\quad\ +\left\{\widetilde{\psi}_a,\widetilde{\psi}_d\right\}_\star\left<\widetilde{\psi}_b\star\widetilde{\psi}_c\star\cdots\star\widetilde{\psi}_f\right>_{\varkappa_\mu}+\dots\\
    &\qquad\quad\ +\left\{\widetilde{\psi}_a,\widetilde{\psi}_f\right\}_\star\left<\widetilde{\psi}_b\star\widetilde{\psi}_c\star\cdots\star\widetilde{\psi}_e\right>_{\varkappa_\mu}-\left<\widetilde{\psi}_b\star\widetilde{\psi}_c\star\cdots\star\widetilde{\psi}_f\star\widetilde{\psi}_a\right>_{\varkappa_\mu}.
\end{align*}
Afterwards, we use the cyclicity of the Grassmann integral in the last expectation value on the r.h.s.\ of the latter expression and the identities
\begin{align*}
 \overline{\psi}_i\star\varkappa_\mu=\mathrm{e}^{q_i}\,\varkappa_\mu\star\overline{\psi}_i\quad\text{and}\quad		
      \psi_i\star\varkappa_\mu=\mathrm{e}^{-q_i}\,\varkappa_\mu\star\psi_i,
\end{align*}
which follow from the fact that $\varkappa_\mu$ is a star product of single states of the form $r_i\overline{\psi}_i\psi_i+1$ and the CAR 
for the star product. Thus, the last expectation value can be written as
\begin{align*}
 \left<\widetilde{\psi}_b\star\widetilde{\psi}_c\star\cdots\star\widetilde{\psi}_f\star\widetilde{\psi}_a\right>_{\varkappa_\mu}
  =\mathrm{e}^{\pm q_a}\left<\widetilde{\psi}_a\star\widetilde{\psi}_b\star\widetilde{\psi}_c\star\cdots\star\widetilde{\psi}_f\right>_{\varkappa_\mu},
\end{align*}
and we conclude with 
\begin{align*}
 &\left<\widetilde{\psi}_a\star\widetilde{\psi}_b\star\widetilde{\psi}_c\star\cdots\star\widetilde{\psi}_f\right>_{\varkappa_\mu}\\
    &\qquad=\frac{\left\{\widetilde{\psi}_a,\widetilde{\psi}_b\right\}_\star}{1+\mathrm{e}^{\pm q_a}}\left<\widetilde{\psi}_c\star\widetilde{\psi}_d\star\cdots\star\widetilde{\psi}_f\right>_{\varkappa_\mu}
      -\frac{\left\{\widetilde{\psi}_a,\widetilde{\psi}_c\right\}_\star}{1+\mathrm{e}^{\pm q_a}}\left<\widetilde{\psi}_b\star\widetilde{\psi}_d\star\cdots\star\widetilde{\psi}_f\right>_{\varkappa_\mu}\\
	&\qquad\quad\ +\frac{\left\{\widetilde{\psi}_a,\widetilde{\psi}_d\right\}_\star}{1+\mathrm{e}^{\pm q_a}}\left<\widetilde{\psi}_b\star\widetilde{\psi}_c\star\cdots\star\widetilde{\psi}_f\right>_{\varkappa_\mu}+\dots\\
	&\qquad\quad\ +\frac{\left\{\widetilde{\psi}_a,\widetilde{\psi}_f\right\}_\star}{1+\mathrm{e}^{\pm q_a}}\left<\widetilde{\psi}_b\star\widetilde{\psi}_c\star\cdots\star\widetilde{\psi}_e\right>_{\varkappa_\mu}.
\end{align*}
We have reduced the expectation value of $2N$ generators to a sum of expectation values of $2(N-1)$ generators. As in \cite{GAU},
the assertion follows by an induction in the number of generators. Finally, the quasifreeness of $\varkappa$
follows from 
\begin{align*}
 \varkappa=\lim\limits_{\mu\to \infty}\frac{\varkappa_\mu}{\int\cD(\overline{\Psi},\Psi)\varkappa_\mu},
\end{align*}
which completes the proof.
\end{proof}

\begin{rem}
Carrying out
the $|M|$-fold star product in $\varkappa_\mu$, we find a more convenient form of $\varkappa_\mu$:
\begin{align*}
\varkappa_\mu=\sum\limits_{Q\subseteq M}\left(-1\right)^{s_Q} \prod\limits_{i\in Q}r_i\prod\limits_{i\in Q}\overline{\psi}_i\prod\limits_{i\in Q}\psi_i
      =\sum\limits_{Q\subseteq M}\left(-1\right)^{s_Q} r_Q\overline{\Psi}_Q\Psi_Q, 
\end{align*}
where $s_Q:=\frac{1}{2}|Q|(|Q|-1)$, $r_Q:=\prod\limits_{i\in Q}r_i$. The sum runs over all ordered subsets $Q\subseteq M$.
\end{rem}


\bibliography{GM_Refs}

\begin{thebibliography}{10}

\bibitem{VBA}
V.~Bach.
\newblock {{E}rror {B}ound for the {H}artree--{F}ock {E}nergy of {A}toms and
  {M}olecules}.
\newblock {\em Communications in Mathematical Physics}, 147(3):527--548, 1992.

\bibitem{BKM}
V.~Bach, H.~K. Kn{\"o}rr, and E.~Menge.
\newblock {Fermion {C}orrelation {I}nequalities {D}erived from {G}- and
  {P}-{C}onditions}.
\newblock {\em Documenta Mathematica}, 17(14):451--481, 2012.

\bibitem{BLS}
V.~Bach, E.~H. Lieb, and J.~P. Solovej.
\newblock {Generalized Hartree--Fock Theory and the Hubbard Model}.
\newblock {\em Journal of Statistical Physics}, 76(1-2):3--89, 1994.

\bibitem{CLS}
E.~Canc{\`e}s, G.~Stoltz, and M.~Lewin.
\newblock {The electronic ground state energy problem: A new reduced density
  matrix approach}.
\newblock {\em The Journal of Chemical Physics}, 125(064101), 2006.

\bibitem{AJC}
A.~J. Coleman.
\newblock {Structure of {F}ermion {D}ensity {M}atrices}.
\newblock {\em Reviews of modern Physics}, 35(3):668--687, 1963.

\bibitem{FRZ}
M.~Combescure and D.~Robert.
\newblock {\em {Coherent States and Applications in Ma\-the\-ma\-ti\-cal
  Physics}}.
\newblock {Theoretical and Mathematical Physics}. Springer-Verlag, 2012.

\bibitem{RME}
R.~M. Erdahl.
\newblock {Representability}.
\newblock {\em International Journal of Quantum Chemistry}, 13(6):697--718,
  1978.

\bibitem{FKT}
J.~Feldman, H.~Kn{\"o}rrer, and E.~Trubowitz.
\newblock {\em {Fermionic Functional Integrals and the Renormalization Group}},
  volume~16 of {\em {CRM Monograph Series}}.
\newblock American Mathematical Society, 2002.

\bibitem{CJP}
C.~Garrod and J.~K. Percus.
\newblock {Reduction of the {$N$}-{P}article {V}ariational {P}roblem}.
\newblock {\em Journal of Mathematical Physics}, 5(12):1756--1776, 1964.

\bibitem{GAU}
M.~Gaudin.
\newblock {Une d{\'e}monstration simplifi{\'e}e du th{\'e}or{\`e}me de Wick en
  m{\'e}chanique statistique}.
\newblock {\em Nuclear Physics}, 15:89--91, 1960.

\bibitem{ELT}
E.~H. Lieb and W.~Thirring.
\newblock {Bound for the {K}inetic {E}nergy of {F}ermions which {P}roves the
  {S}tability of {M}atter}.
\newblock {\em Physical Review Letters}, 35(11):687--689, 1975.
\newblock Errata 35, 1116 (1975).

\bibitem{POL}
P.-O. L{\"o}wdin.
\newblock {{Q}uantum {T}heory of {M}any-{P}article {S}ystems. {I}. {P}hysical
  {I}nterpretations by {M}eans of {D}ensity {M}atrices, {N}atural
  {S}pin-{O}rbitals, and {C}onvergence {P}roblems in the {M}ethod of
  {C}onfigurational {I}nteraction}.
\newblock {\em Physical Review}, 97(6):1474--1489, 1955.

\bibitem{DAM}
D.~A. Mazziotti.
\newblock {Variational minimization of atomic and molecular ground-state
  energies via the two-particle reduced density matrix}.
\newblock {\em Physical Review A}, 65(062511), 2002.

\bibitem{MAZ}
D.~A. Mazziotti.
\newblock {Structure of {F}ermionic {D}ensity {M}atrices: {C}omplete
  {$N$}-{R}epresentability {C}onditions}.
\newblock {\em Physical Review Letters}, 108(263002), 2012.

\bibitem{WSP}
{W. d. S.} Pedra.
\newblock {\em {Zur mathematischen Theorie der Fermifl{\"u}ssigkeiten bei
  positiven Temperaturen}}.
\newblock PhD thesis, Universit{\"a}t Leipzig, 2005.

\bibitem{MSH}
M.~Salmhofer.
\newblock {\em {Renormalization --- An Introduction}}.
\newblock Springer-Verlag, 1998.

\bibitem{LAT}
L.~A. Takhtajan.
\newblock {\em {Quantum Mechanics for Mathematicians}}, volume~95 of {\em
  {Graduate Studies in Mathematics}}.
\newblock American Mathematical Society, 2008.

\bibitem{WTH}
W.~Thirring.
\newblock {\em {Quantenmechanik gro{\ss}er Systeme}}, volume~4 of {\em
  {Lehrbuch der Mathematischen Physik}}.
\newblock Springer-Verlag, 2008.

\bibitem{ZBF}
Z.~Zhao, B.~J. Braams, M.~Fukuda, M.~L. Overton, and J.~K. Percus.
\newblock {The reduced density matrix method for electronic structure
  calculations and the role of three-index representability conditions}.
\newblock {\em Journal of Chemical Physics}, 120(2095), 2004.

\end{thebibliography}

\end{document}